\documentclass{amsart}
\usepackage{graphicx}

\newtheorem{theorem}{\bf Theorem}[section]

\newtheorem{lemma}{\bf Lemma}[section]

\newtheorem{proposition}{\bf Proposition}[section]

\usepackage{bm}
\usepackage{amsmath,amsthm,amssymb,stmaryrd}
\usepackage{enumerate}

\newcommand{\bbint}[2]{\ensuremath{\;\diagdown\!\!\!\!\diagdown\!\!\!\!\!\!\int_{#1}^{#2}}}

\begin{document}

\title[Term by term integration]{The problem of missing terms in term by term integration involving divergent integrals}

\author{Eric A. Galapon}
\address{Theoretical Physics Group, National Institute of Physics, University of the Philippines, Diliman Quezon City, 1101 Philippines}





\begin{abstract}
	Term by term integration may lead to divergent integrals, and naive evaluation of them by means of, say, analytic continuation or by regularization or by the finite part integral may lead to missing terms. Here, under certain analyticity condition, the problem of missing terms for the incomplete Stieltjes transform, $\int_0^a f(x) (\omega+x)^{-1} \mathrm{d}x$, and the Stieltjes transform itself, $\int_0^{\infty} f(x) (\omega+x)^{-1} \mathrm{d}x$, is resolved by lifting the integration in the complex plane.  It is shown that the missing terms arise from the singularities of the complex valued function $f(z) (\omega + z)^{-1}$, with the divergent integrals arising from term by term integration interpreted as finite part integrals. 
\end{abstract}

	
\maketitle 

\section{Introduction}
Term by term integration is one of the frequently used methods in evaluating integrals and in constructing asymptotic expansions of integrals. It typically involves expanding the integrand or a factor of it and then interchanging the order of summation and integration. Of course the order  cannot be arbitrarily interchanged without satisfying certain uniformity conditions. But we, nevertheless, often chivalrously perform the interchange at the cost of introducing divergent integrals in an otherwise well defined integral. We often see this happen in asymptotic evaluation of integrals where the first step in constructing an asymptotic expansion is by deriving a formal expansion, say, by means of an appropriate expansion of the integrand followed by a term be term integration \cite{wong}. The emergence of divergent integrals in such formal manipulation calls for their proper interpretation and treatment. The interpretation of divergent integrals is made complicated by the fact that a divergent integral may be assigned different values corresponding to different ways of interpreting the divergent integral. A divergent integral can be assigned meaningful value by means of analytic continuation \cite{analytic},  Caesaro limits \cite{caesaro}, regularization methods \cite{regularization}, distribution theory \cite{distribution}, or finite part integration \cite{monegato} or still by some other means \cite{others1,others2,others3}. These generally assign different values to the divergent integral. A naive application of any of these interpretations of a divergent  integral in term by term integration can lead to misleading results, sometimes yielding the correct value but more often correctly reproducing only some of the terms of the correct results but completely missing out some group of terms.

The problem with naive term by term integration is exemplified by the asymptotic evaluation of the Stieltjes transform $\int_0^{\infty} h(x)(\omega+x)^{-1}\mathrm{d}x$. The canonical example is given by the integral \cite{wong}
\begin{equation}\label{example}
\int_0^{\infty} \frac{1}{\sqrt[3]{1+x}\, (\omega+x)}\mathrm{d}x  .
\end{equation} 
Inserting the binomial expansion of $(1+x)^{-1/3}$ at infinity and performing term by term integration lead to the infinite series
\begin{equation}
\int_0^{\infty} \frac{1}{\sqrt[3]{1+x}\, (\omega+x)}\mathrm{d}x\leadsto \sum_{s=0}^{\infty} \binom{-1/3}{s} \int_0^{\infty} \frac{x^{-s-1/3}}{(\omega+x)} \mathrm{d}x .
\end{equation} 
Only the $s=0$ term exists and the rest of the integrals do not exist due to the non-integrable singularity at the origin. The question now arises as to how the divergent integrals should be interpreted.

The divergent integrals can be interpreted in several ways. One way is to interpret them via analytic continuation. The idea behind ``analytic continuation is to match the divergent function of interest to an analytic function in a region where the original function is well behaved'' \cite{analytic}. For instance, we have the integral
\begin{equation}\label{analex}
\int_0^{\infty} \frac{x^{-\mu}}{(\omega+x)} \mathrm{d}x = \frac{\pi\csc\pi\mu}{\omega^{\mu}}, \;\;\; 0<\mathrm{Re}\, \mu<1 .
\end{equation}
The integral in the left hand side is defined only in the strip $0<\mathrm{Re}\, \mu <1$, and, in this strip, equal to the right hand side. However, the right hand side is well defined beyond the strip of analyticity of the integral so that it is the (unique) analytic extension of the integral in the entire complex plane. By analytic continuation, the divergent integrals are then assigned the values equal to the values of the right hand side of equation (\ref{analex}) at $\mu=s+1/3$ for every positive integer $s$. Applying this interpretation, term by term integration leads to the expansion
\begin{equation} \frac{2\pi}{\sqrt{3}} \sum_{s=0}^{\infty} (-1)^s \binom{-1/3}{s} \frac{1}{\omega^{s+1/3}}\label{wrong} .
\end{equation} 
However, the given integral (\ref{example}) can be evaluated explicitly and the result expanded for $\omega>1$. The result is given by \cite{distribution,mcwong}
\begin{equation} \frac{2\pi}{\sqrt{3}}\sum_{s=0}^{\infty} (-1)^s \binom{-1/3}{s} \frac{1}{\omega^{s+1/3}} - \sum_{s=1}^{\infty} \frac{3^s (s-1)!}{2\cdot 5 \cdots (3s-1)}\frac{1}{\omega^s},\;\;\; \omega>1\label{correct} .
\end{equation}

Comparing expressions (\ref{wrong}) and (\ref{correct}), we find that expression (\ref{wrong}) correctly reproduces the first group of terms of (\ref{correct}) but completely misses out the second group. This demonstrates how naive term by term integration involving divergent integrals can lead to missing terms. This problem for the Stieltjes transform was first resolved by McClure and Wong by interpreting the Stieltjes transform as a linear functional over some fundamental space of test functions, rendering in the process every expression involved in the expansion of the integrand as a tempered distribution, and assigning interpretation to the divergent integrals as functionals over test functions \cite{mcwong}.  
By establishing the precise relationship among the distributions, McClure and Wong successfully reproduced the missing terms. There is, however, one curious feature of the approach. It requires the evaluation of a Mellin integral transform which may be divergent for all relevant values of the variable of the Mellin transform; and when it happens that the integral transform is divergent, it is assigned the value equal to the analytic extension of the Mellin transform. Moreover, in the intermediate steps leading to the Mellin transform arise divergent integrals which were evaluated by means of analytic continuation. However, it is not transparent why analytic continuation is justifiable in those instances. But since the canonical example discussed above already demonstrates that analytic continuation may fail, a rigorous justification of the assignment of the divergent Mellin integral and the divergent integrals in the intermediate steps by their analytic continuations is desirable. 

In this paper we give term by term integration another look and the interpretation of divergent integrals that accompanies in such operation without appealing to distribution theory and to analytic continuation. Here we will consider the incomplete Stieltjes transform 
\begin{equation}\label{incomplete}
	S_a(\omega)=\int_0^a \frac{ h(x)}{\omega+x}\, \mathrm{d}x,\;\;\; \omega>0 , \;\;\; a>0 ,
\end{equation}
from which the well-known Stieltjes transform can be recovered in the limit as $a\rightarrow\infty$, provided the limit exists.We will consider two expansions of the integrand leading to divergent integrals upon term by term integration, which are
\begin{equation}\label{expand1}
	\frac{1}{\omega+x}=\sum_{j=0}^{\infty} (-1)^j \frac{\omega^j}{x^{j+1}}
	\end{equation}
	\begin{equation}
	\label{expand2}
	h(x)=\sum_{j=0}^{\infty} \frac{c_j}{x^{j+\alpha}}
\end{equation}
Expansion (\ref{expand1}) is appropriate for term by term integration for small $\omega$ or at the origin; and expansion (\ref{expand2}) is for term by term integration for large $\omega$ or at infinity.  Both expansions lead to the family of divergent integrals
\begin{eqnarray}\label{divergent}
	\int_0^a \frac{f(x)}{x^{m +\nu}}\mathrm{d}x, \;\;\; m=1, 2, \dots, \;\; 0\leq \nu <1, \;\; a>0.  
\end{eqnarray}
for some locally integrable function $f(x)$ with $f(0)\neq 0$. Term by term integration involving expansion (\ref{expand2}) has been the area of asymptotics and the area of application of the distributional approach due to McClure and Wong \cite{distribution,mcwong}. On the other hand, it is for the first time that term by term integration involving divergent integrals arising from expansion (\ref{expand1}) is dealt with. 

Here we will approach the problem of dealing with the divergent integrals arising from term by term integration by assigning them values equal to their finite parts. Roughly the finite part of a divergent integral is obtained by temporarily modifying the integral to become convergent, followed by identifying the terms that diverge as the modified integral approaches the divergent integral; the finite part is what remains after dropping the diverging terms. The finite part is conceptually distinct from the value assigned by analytic continuation (or the analytic value in short). However, the analytic value may be equal to the finite part but it may develop poles where the finite part is well-defined \cite{monegato2} so that the finite part is more robust than analytic continuation.

Now a fundamental insight gleaned from the distributional approach of McClure and Wong is that it is not enough to assign an interpretation to a divergent integral without establishing the precise relationship between the assigned values of the divergent integrals to the given convergent integral, which is the Stietljes transform in our case. This applies to our choice of the finite part. The divergent integrals in the above canonical example can be shown to have finite parts equal to their analytic values. Thus the finite part already fails to give the correct result without further consideration. The problem now is to discover the precise rules of using the finite part of a divergent integral in evaluating a convergent integral such as the Stieltjes integral transform. The are two key hints that point us the way to proceed. The first hint is provided by our recent work that saw the representation of the Cauchy principal value and the Hadamard finite part of the generalized principal value by Fox \cite{fox} as complex contour integrals \cite{galapon}. The second hint is the known fact that an integral in the real line can be cast into a contour integral in the complex plane \cite{costin}. Then the emergent divergent integrals (\ref{divergent}) can be interpreted as a finite part integral provided an appropriate contour integral representation of the finite part is obtained, and provided the  original Stieltjes integral is likewise represented as a contour integral consistent with the contour integral representation of the finite part integral. We will find that it is precisely the common contour in the complex contour integral representations of the finite part and the incomplete Stieltjes transform that bridges them and that leads to the meaningful use of the finite part integral.

In this paper, we will obtain the contour integral representation of the finite part of the divergent integrals (\ref{divergent}), from which we read off the required consistent contour integral representation of the incomplete Stieltjes transform. Replacing the integral \ref{incomplete} with its equivalent contour integral representation in the complex plane and performing the term by term integration over this representation will lead us to uniquely identify the divergent integrals that arise as finite part integrals. Moreover, we will find that the missing terms emerge as contributions coming from the complex singularities of the function $h(z)(\omega+z)^{-1}$ (which is the extension of the integrand $h(x)(\omega+x)^{-1}$ in integral (\ref{incomplete}) in the complex plane). Here we will primarily treat term by term integration arising from expansion (\ref{expand1}), and show that the finite part integral in complex contour integral representation leads to an analytic (exact) evaluation of the incomplete Stieltjes transform and the Stieltjes transform in positive powers of $\omega$. We will find, in particular, that the missing term arise from the simple pole of $(\omega+z)^{-1}$.  On the other hand, we will cursorily consider term by term integration arising from expansion (\ref{expand2}), contenting ourselves with an application of our method to the specific example of a generalization of the canonical example, and demonstrating that our approach reproduces a known result of the distributional approach to asymptotics. We will find that, for the given example, the missing terms arise from the branch point of $h(z)$. We will treat elsewhere an in depth comparison of our approach to that of the distributional approach of McClure and Wong.  

Before we proceed, we acknowledge that other methods may be used in obtaining the same expansion for either large or small $\omega$ without the use of divergent integrals. It may even be the case that they are more straightforward than our treatment here. However, it is not our objective to develop a method of evaluating the Stieltjes transform that is deemed superior to all other alternatives. Instead it is our objective to address the fundamental question of the origin and the recovery of missing terms arising from term by term integration of an infinite series that leads to an infinite series whose terms are divergent integrals. This is not just a formal mathematical problem, but it is also a problem whose elucidation may lead to deeper insights into perturbative solutions of physical problems involving infinite series of divergent integrals \cite{analytic,regularization,renormalization}. The case of missing terms arising from a mere assignment of a finite value to a divergent integral in an infinite series rises the specter of missing terms in such a perturbative solution. A method that does not involve divergent integrals will bring little, if there is any, direct insight into the problem of missing terms. But a method that explicitly involves divergent integrals, such as ours here, will deliver new perspective into the use of divergent integrals in such perturbative solutions.

The paper is organized as follows. In Section (\ref{contour}) we obtain the contour integral representation of the finite part of the divergent integrals (\ref{divergent}) for pole ($\nu=0$) and branch point ($\nu\neq 0$) singularities. In Section (\ref{origin}), we perform term by term integration at the origin corresponding to the expansion (\ref{expand1}) for both cases of singularities involving entire integrands; we derive the known series expansions of the exponential function and the incomplete gamma function as examples. In Section (\ref{origin2}), we consider expansion at the origin for holomorphic integrands. In Section (\ref{infinity}) we perform term by term integration at infinity corresponding to the expansion (\ref{expand2}) for a specific class of integrals, the canonical example (\ref{example}) is a special case of which. Through out the paper, by integrable we shall mean Riemann integrable.

\section{Contour integral representation of finite part integrals of divergent integrals with  singularities at the origin}\label{contour}

We now obtain the contour integral representation of the finite part of the divergent integral (\ref{divergent}) under certain conditions to be detailed below. Assuming that the function $f(x)$ is integrable in the interval $[0,a]$, the finite part integral is obtained by replacing the lower limit of integration by some positive $\epsilon<a$. The resulting integral is then grouped in two terms,
\begin{equation}
\int_{\epsilon}^{a} \frac{f(x)}{x^{m+\nu}} \mathrm{d}x = C_{\epsilon} + D_{\epsilon}
\end{equation} 
where $C_{\epsilon}$ is the group of terms that posses a finite value in the limit as $\epsilon\rightarrow 0$, and $D_{\epsilon}$ is the group of terms that diverge in the same limit. The finite part of the divergent integral is obtained by dropping the diverging term $D_{\epsilon}$ and assigning the limit of $C_{\epsilon}$ as the value of the divergent integral, otherwise known as the finite part integral (FPI) \cite{kanwal,hadamard}, that is
\begin{equation}\label{limit1}
\mathrm{FPI}\!\int_0^{a} \frac{f(x)}{x^{m+\nu}}\mathrm{d}x = \lim_{\epsilon\rightarrow 0} C_{\epsilon}.
\end{equation}
Equivalently the finite part integral can be expressed in the form
\begin{equation}\label{limit2}
\mathrm{FPI}\!\int_0^{a} \frac{f(x)}{x^{m+\nu}}\mathrm{d}x  = \lim_{\epsilon\rightarrow 0} \left[\int_{\epsilon}^{a} \frac{f(x)}{x^{m+\nu}} \mathrm{d}x - D_{\epsilon}\right] .
\end{equation}
These different representations of the finite part integral will serve us separate purposes. We will use equation (\ref{limit1}) to compute explicitly the finite part integral later; while we will use equation (\ref{limit2}) to establish the desired contour integral representation of the finite part integral. 

Now given the function $f(x)$, the complex extension of $f(x)$ is the function, $f(z)$, obtained from $f(x)$ itself by replacing the real variable $x$ with the complex variable $z$.  Here, as in \cite{galapon}, we assume that the function $f(x)$ has a complex extension, $f(z)$, that is analytic in some region of the complex plane containing the strip $[0,a]$. We will denote the finite part of the divergent integral (\ref{divergent}) under this assumption by
$\bbint{_0}{a}$. We distinguish this with the more familiar notation $\mathrm{FPI}\int_0^a$ because the function $f(x)$ does not not necessarily have an analytic extension in the complex plane for the finite part integral to exist. However, this is not much of a restriction because the functions $f(x)$ involved in practice have analytic extensions.

The contour integral representation of the divergent integral (\ref{divergent}) will depend on the value of $\nu$. The cases $\nu=0$ (pole singularity) and $\nu\neq 0$ (branch point singularity) are distinct and have to be considered separately. The representation of the finite part integral given by equation (\ref{limit1}) or (\ref{limit2}) by a contour integral in the complex plane means that the FPI is the value of an absolutely convergent integral. This lifts the somewhat vague meaning of the finite part, and endows it with an unambiguous interpretation as a legitimate integral itself.

\subsection{Pole singularity}
\begin{proposition}
	Let $f(x)$ be $(n+1)$ continuously differentiable in the interval $[0,a]$ and $M=\sup\{\left|f^{(n+1)}(x)\right|, \, x\in[0,a]\}<\infty$. If $f(0)\neq 0$, then we have the finite part integral
	\begin{equation}\label{fpi1}
	\mathrm{FPI}\int_0^a \frac{f(x)}{x^{n+1}} \mathrm{d}x = \lim_{\epsilon\rightarrow 0^+}\left[\int_{\epsilon}^a \frac{f(x)}{x^{n+1}} -\sum_{k=0}^{n-1}\frac{f^{(k)}(0)}{k! (n-k)}\frac{1}{\epsilon^{n-k}} + \frac{f^{(n)}(0)}{n!}\ln \epsilon\right] ,
	\end{equation}
	for all $n=0, 1, 2, \dots$.
\end{proposition}
\begin{proof}
	For some $\epsilon\in (0,a)$ we consider the convergent integral $\int_{\epsilon}^{a}f(x) x^{-(n+1)} \mathrm{d}x$. By the differentiability of $f(x)$ up to the $(n+1)$-th order we have the expansion about $x=0$,
	\begin{equation}
	f(x)=\sum_{k=0}^{n} \frac{f^{(k)}(0)}{k!} x^k +R_{n+1}(x) .
	\end{equation}
	Inserting this expansion and performing the integral yields
	\begin{eqnarray}
	\int_{\epsilon}^{a} \frac{f(x)}{x^{n+1}}\mathrm{d}x &=& - \sum_{k=0}^{n-1} \frac{f^{(k)}(0)}{k! (n-k)} \left(\frac{1}{a^{n-k}} - \frac{1}{\epsilon^{n-k}}\right)  + \frac{f^{(n)}(0)}{n!} \left(\ln a -\ln \epsilon \right) \nonumber\\
	&& \hspace{24mm} + \int_{\epsilon}^{a}\frac{R_{n+1}(x)}{x^{n+1}} \mathrm{d}x .
	\end{eqnarray}
	
	Given that $|R_{n+1}(x)|\leq M |x|^{n+1}/(n+1)!$, for some constant $M>0$, we obtain the following bound for the integral over the remainder term,
	\begin{eqnarray}
	\left|\int_{\epsilon}^a \frac{R_{n+1}(x)}{x^{n+1}} \mathrm{d}x\right|\leq \int_{\epsilon}^a \frac{\left|R_{n+1}(x)\right|}{x^{n+1}} \mathrm{d}x \leq \frac{M}{(n+1)!} (a-\epsilon) \leq \frac{M a}{(n+1)!}.
	\end{eqnarray}
	Thus the remainder integral exists in the limit as $\epsilon\rightarrow 0$. Collecting all the terms that diverge individually as $\epsilon\rightarrow 0$ in one side, we have
	\begin{eqnarray}
	\int_{\epsilon}^{a} \frac{f(x)}{x^{n+1}}\mathrm{d}x -\sum_{k=0}^{n-1} \frac{f^{(k)}(0)}{k! (n-k)} \frac{1}{\epsilon^{n-k}} + \frac{f^{(n)}(0)}{n!} \ln \epsilon && \nonumber \\
	&&  \hspace{-60mm} = - \sum_{k=0}^{n-1} \frac{f^{(k)}(0)}{k! (n-k)}\frac{1}{a^{n-k}}  + \frac{f^{(n)}(0)}{n!} \ln a  
	 + \int_{\epsilon}^{a}\frac{R_{n+1}(x)}{x^{n+1}} \mathrm{d}x .\label{limit}
	\end{eqnarray}
	The limit of the right hand side of the equation (\ref{limit}) exists as $\epsilon\rightarrow 0$, and so must be the left hand side in the same limit. Finally taking the limit in equation \ref{limit} as $\epsilon\rightarrow 0$ yields the finite part integral  (\ref{fpi1}).
\end{proof}

The second and third terms in the left hand side of equation (\ref{limit}) are exactly the diverging terms that have to be removed from the diverging first term to obtain a finite limit, or the diverging terms that have to be dropped to obtain the Hadamard finite of the integral.

\begin{figure}
	\includegraphics[scale=0.4]{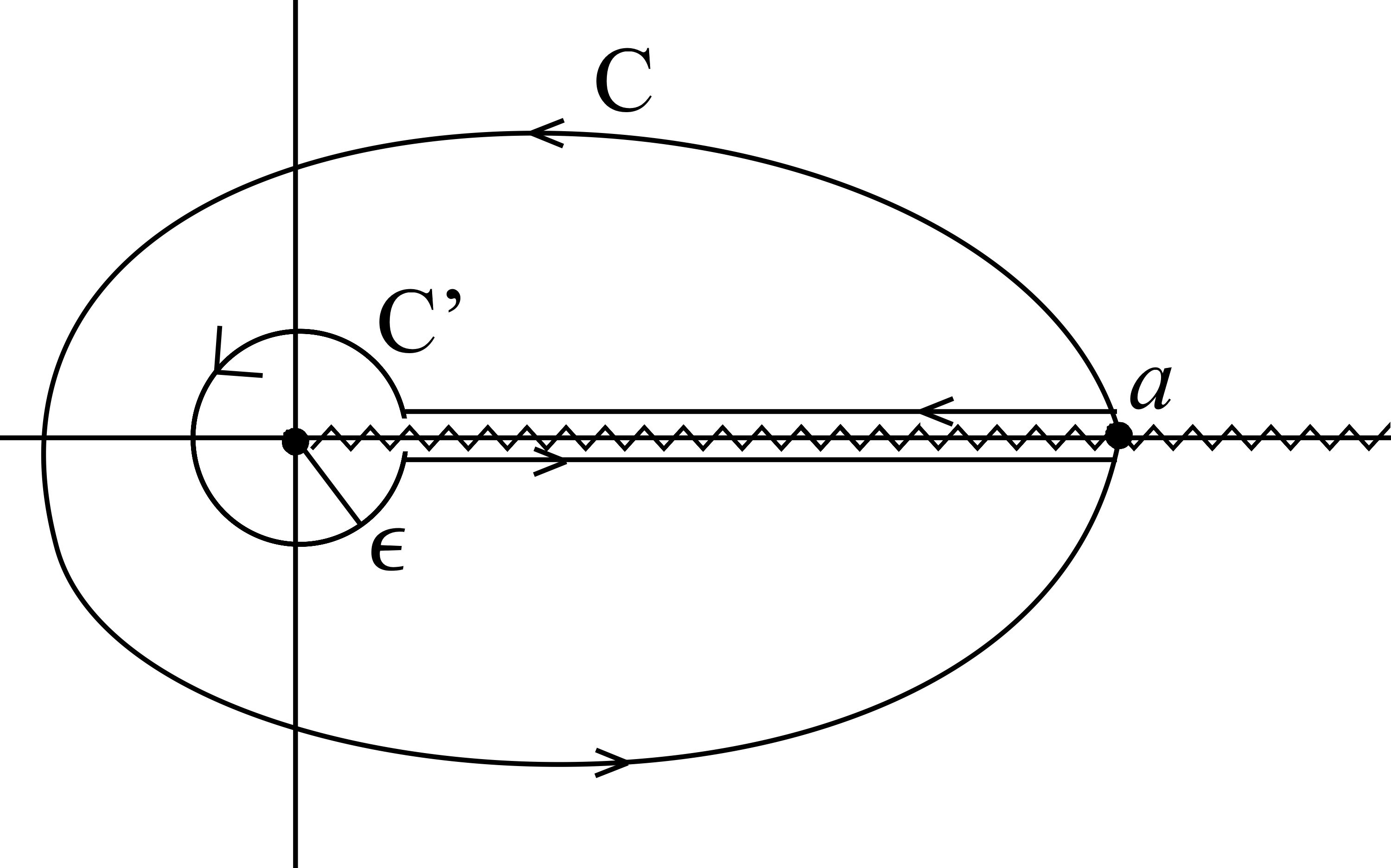}
	\caption{The contour of integration. The contour $\mathrm{C}$ does not enclose any pole of $f(z)$.}
	\label{tear}
\end{figure}

\begin{theorem}\label{prop1}
	Let the complex extension, $f(z)$, of $f(x)$ be analytic in some neighborhood of the interval $[0,a]$. If $f(0)\neq 0$, then 
	\begin{equation}\label{result1}
	\bbint{0}{a}\frac{f(x)}{x^{n+1}}\mathrm{d}x=\frac{1}{2\pi i}\int_{\mathrm{C}} \frac{f(z)}{z^{n+1}} \left(\log z-\pi i\right)\mathrm{d}z, \;\; n=0, 1, \dots 
	\end{equation}
	where $\log z$ is the complex logarithm whose branch cut is the positive real axis and $\mathrm{C}$ is the contour straddling the branch cut of $\log z$ starting from $a$ and ending at $a$ itself, as depicted in Figure-\ref{tear}.
\end{theorem}
\begin{proof}
	Deform the contour $\mathrm{C}$ into the contour $\mathrm{C}'$ shown in Figure-\ref{tear}, and evaluate the integral along that contour. Then
	\begin{equation}
	\int_{\mathrm{C}} \frac{f(z) \log z}{z^{n+1}} \mathrm{d}z=\int_a^{\epsilon} \frac{f(x) \ln x}{x^{n+1}} \mathrm{d}x + \int_{\epsilon}\frac{f(z) \log z}{z^{n+1}} \mathrm{d}z + \int^a_{\epsilon} \frac{f(x) \log\left(x \mathrm{e}^{2\pi i}\right)}{x^{n+1}} \mathrm{d}x \label{ke}
	\end{equation}
	where the first term comes from the upper part of the contour, the second term from the small circle about the origin, and the third term from the lower part of the contour. With $\log(x\mathrm{e}^{2\pi i})= \ln x + 2\pi i$, equation (\ref{ke}) simplifies to
	\begin{equation}\label{integ1}
	\int_{\mathrm{C}} \frac{f(z) \log z}{z^{n+1}} \mathrm{d}z=2\pi i \int_{\epsilon}^{a} \frac{f(x)}{x^{n+1}} \mathrm{d}x + \int_{\epsilon}\frac{f(z) \log z}{z^{n+1}} \mathrm{d}z .
	\end{equation}
	
	We now evaluate the integral along the small circle, the second term in equation \ref{integ1}. Since, by assumption, $f(z)$ is analytic at $z=0$, we have the expansion
	\begin{equation}
	f(z)=\sum_{k=0}^n \frac{f^{(k)}(0)}{k!} z^k + \mathcal{O}(z^{n+1}), \;\; n=0, 1, \dots. 
	\end{equation}
	Substituting this back into the integral and using the parametrization $z=\epsilon \mathrm{e}^{i\theta}$, $0\leq \theta<2\pi$, we obtain
	\begin{eqnarray}
	\int_{\epsilon} \frac{f(z) \log z}{z^{n+1}} \mathrm{d}z &=& \sum_{k=0}^n \frac{f^{(k)}(0)}{k!} \frac{1}{\epsilon^{n-k}} \left[i \ln\epsilon \int_0^{2\pi} \mathrm{e}^{-i (n-k) \theta} \mathrm{d}\theta - \int_0^{2\pi} \mathrm{e}^{-i(n-k)\theta} \theta \mathrm{d}\theta\right]\nonumber \\
	&& \hspace{18mm}+ \mathcal{O}(\epsilon),
	\end{eqnarray}
	where we have used $\log(\epsilon \mathrm{e}^{i\theta}) = \ln\epsilon + i\theta$. The integrals are evaluated with the following integrals
	\begin{equation}
	\int_0^{2\pi} \mathrm{e}^{-i (n-k)\theta}\mathrm{d}\theta=\left\{\begin{array}{cc}
	0 &, n\neq k\\
	2\pi &, n=k
	\end{array}
	\right. ,
	\end{equation}
	\begin{equation}
	\int_0^{2\pi} \mathrm{e}^{-i (n-k)\theta}\theta \mathrm{d}\theta=\left\{\begin{array}{cc}
	\frac{2\pi i}{(n-k)} &, n\neq k\\
	2\pi^2 &, n=k
	\end{array}
	\right. .
	\end{equation}

Substituting the above integrals in equation (\ref{integ1}), the integral along the entire contour $\mathrm{C}$ assume the form
\begin{eqnarray}
\int_{\mathrm{C}} \frac{f(z) \log z}{z^{n+1}} \mathrm{d}z &=& 2\pi i \left[\int_{\epsilon}^{a}\frac{f(x)}{x^{n+1}} \mathrm{d}x -  \sum_{k=0}^{n-1} \frac{f^{(k)}(0)}{k!(n-k)}\frac{1}{\epsilon^{n-k}} + \ln\epsilon \frac{f^{n}(0)}{n!}\right] \nonumber \\
&& \hspace{24mm}- 2 \pi^2 \frac{f^{(n)}(0)}{n!}+ \mathcal{O}(\epsilon) \label{kwe}
\end{eqnarray}
We can already discern the Hadamard finite part from the group of first terms in the limit as $\epsilon\rightarrow 0$. To arrive at the final form, we use the Cauchy integral formula to rewrite the derivative term. Since the contour $\mathrm{C}$ is closed, by the Cauchy integral formula we have
\begin{equation}
f^{(n)}(0)= \frac{n!}{2\pi i} \int_{\mathrm{C}} \frac{f(z)}{z^{n+1}} \mathrm{d}z\label{kwe2}
\end{equation}
We substitute equation \ref{kwe2} back into equation \ref{kwe}, divide the entire expression by $2\pi i$ and then take the limit as $\epsilon\rightarrow 0$ leads us to the result (\ref{result1}).
\end{proof}

\subsection{Branch Point Singularity}
\begin{proposition}
	Let $f(x)$ be $m$ continuously differentiable in the interval $[0,a]$ and $M=\sup\{\left|f^{(m)}(x)\right|, \, x\in[0,a]\}<\infty$. If $f(0)\neq 0$, then we have the finite part integral
	\begin{equation}\label{branchgen}
	\mathrm{FPI}\int_0^a \frac{f(x)}{x^{m+\nu}} \mathrm{d}x = \lim_{\epsilon\rightarrow 0^+}\left[\int_{\epsilon}^{a}\frac{f(x)}{x^{m+\nu}}\mathrm{d}x - \sum_{j=0}^{m-1} \frac{f^{(j)}(0)}{j! (n+\nu-j-1)}\frac{1}{\epsilon^{n+\nu-j-1}}\right]
	\end{equation}
	for all $m=1, 2, \dots$ and $\nu\in(0,1)$.
\end{proposition}
\begin{proof}
	The term proportional to $x^m$ in the Taylor expansion of $f(x)$ leads to the integral $\int_{\epsilon}^{a} x^{-\nu} g(x)\mathrm{d}x$ with some $g(x)$ integrable in $[0,a]$. This term has a finite limit as $\epsilon\rightarrow 0$. We then expand $f(x)$ about $x=0$ up to the order $(m-1)$ and then proceed in the same manner as we have done in Theorem-\ref{prop1} to establish the finite part integral (\ref{branchgen}). 
\end{proof}

\begin{theorem}
	Let the complex extension, $f(z)$, of $f(x)$ be analytic in some neighborhood of the interval $[0,a]$. If $f(0)\neq 0$, then
	\begin{equation}\label{branch}
	\bbint{0}{a} \frac{f(x)}{x^{m+\nu}}\mathrm{d}x = \frac{1}{\left(\mathrm{e}^{-2\pi \nu i } -1\right)} \int_{\mathrm{C}} \frac{f(z)}{z^{m+\nu}} \mathrm{d}z, \;\;\; m=1, 2, \dots, \;\;\; 0<\nu<1,
	\end{equation}
	where the branch of $z^{-\nu}$ is such that it is positive on top of the positive real axis  and the contour $\mathrm{C}$ is the contour straddling the branch cut of $z^{-\nu}$ starting from $a$ and ending at $a$ itself, as depicted in Figure-1.
\end{theorem}
\begin{proof} 
	We proceed in the same manner as we did in Theorem \ref{prop1} to establish the equality \ref{branch}.
\end{proof}

\subsection{Example} Let us compute the finite part integral of the divergent integral $\int_0^a x^{-n-\nu} \mathrm{d}x$ for $n=1,2,\dots$, $0<\nu<1$ and $a>0$. For $\epsilon>0$, we have
\begin{equation}
\int_{\epsilon}^{a}\frac{1}{x^{n+\nu}}\mathrm{d}x = -\frac{1}{(n+\nu-1)} \frac{1}{a^{n+\nu-1}}+ \frac{1}{(n+\nu-1)} \frac{1}{\epsilon^{n+\nu-1}} .
\end{equation}
The second term diverges as $\epsilon\rightarrow 0$ so we drop it, leaving only the first term to give the finite part integral
\begin{equation}
\mathrm{FPI}\!\int_{\epsilon}^{a}\frac{1}{x^{n+\nu}}\mathrm{d}x = -\frac{1}{(n+\nu-1)}\frac{1}{a^{n+\nu-1}} .
\end{equation}
This exactly the same result obtained using equation (\ref{branchgen}). To use the integral representation, we choose the unit circle centered at the origin as our contour of integration. With the parametrization $z=a \mathrm{e}^{i\theta}$, we obtain
\begin{eqnarray}
\mathrm{FPI}\!\int_0^a \frac{1}{x^{n+\nu}}\, \mathrm{d}x &=& \frac{1}{\left(\mathrm{e}^{-2\pi \nu i}-1\right)} \int_0^{2\pi} \frac{1}{a^{n+\nu}\mathrm{e}^{i (n+\nu)\theta}} i a \mathrm{e}^{i\theta} \mathrm{d}\theta \nonumber \\
&=& -\frac{1}{(n+\nu-1)}\frac{1}{a^{n+\nu-1}} . \nonumber
\end{eqnarray}
Any contour continuously deformable to the circle will yield the same value of the finite part integral.

\subsection{Remark}
It is imposed above that the function $f(x)$ does not vanish at the origin so that the origin is necessarily non-integrable. However, it may happen that $f(x)$ vanishes there, in particular, $f(x)=x^m g(x)$, where $g(0)\neq 0$, for some positive integer $m$. Then the integral $\int_0^a f(x) x^{-m-\nu} \mathrm{d}x$   exists for all $0\leq \nu<1$, and its value is given by $\int_0^a g(x) x^{-\nu} \mathrm{d}x$. Provided the complex extension function $g(z)$ has the desired analytic properties, it is straightforward to show that equations (\ref{result1}) and (\ref{branch}) yield the respective values $\int_0^a g(x)\mathrm{d}x$ (for $\nu=0$) and $\int_0^a g(x) x^{-\nu}\mathrm{d}x$ (for $0<\nu<1$). That is the finite part integral and the regular integral (e.g. Riemann integral) coincide when the given integral $\int_0^a f(x) x^{-m-\nu} \mathrm{d}x$ converges absolutely.  This is expected as the finite part integral of an absolutely convergent integral is just the value of the integral itself, i.e. the integral has no diverging component. This observation allows us to remove the restriction $f(0)\neq 0$ and interpret all integrals that appear as finite part integral wherever appropriate.

\section{Term by term integration at the origin for functions with entire extensions}\label{origin}

\begin{figure}
	\includegraphics[scale=0.4]{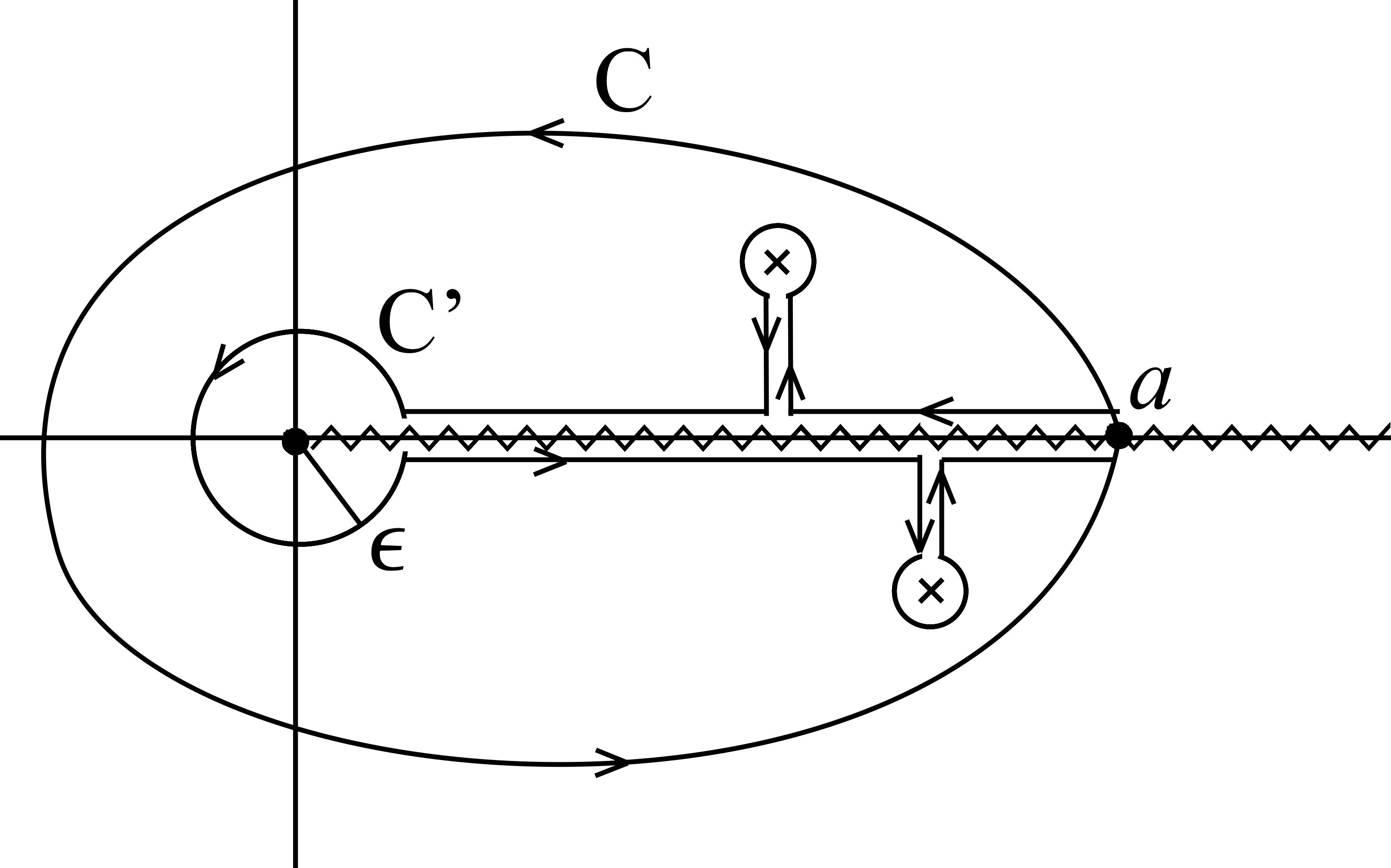}
	\caption{The contour of integration in the contour representation of a real integral. The contour $\mathrm{C}$ encloses some or all the poles of the integrand.}
	\label{fig:boat1}
\end{figure}

Here and in the succeeding section, we will consider the term by term integration of the incomplete Stieltjes transform  
\begin{equation}\label{incom}
\int_0^a \frac{ x^{-\nu} f(x)}{ \omega+x}\mathrm{d}x, \;\;\; 0\leq \nu<1, 
\end{equation}
  in the neighborhood of $\omega=0$. The overarching idea behind our approach to the missing terms is to obtain a contour integral representation of the integral (\ref{incom}) that is consistent with the contour integral representation of the finite part integral. By consistent we mean that the same contour used in the contour integral representation of the finite part integral is the same contour used in the contour integral representation of the incomplete Stieltjes transform. The representation will be dictated by the value of $\nu$, either $\nu=0$ or $\nu\neq 0$, and the analytic behavior of $f(z)$ in the complex plane. In this section, we will consider the case in which $f(x)$ has an entire complex extension $f(z)$. 

\subsection{Case 1:} $\nu =0$
\begin{lemma} \label{lemma0}
	Let $g(x)$ be integrable in the interval $[0,a]$ and that its complex extension $g(z)$ is analytic in a region containing the interval $[0,a]$ in its interior and holomorphic elsewhere. Then
	\begin{equation}\label{contourep1}
	\int_0^a g(x) \, \mathrm{d}x = \frac{1}{2\pi i} \int_{\mathrm{C}} g(z) \log z \, \mathrm{d}z - \sum_k \mathrm{Res}\left[\log z \, g(z)\right]_{z_k},
	\end{equation}
	where the branch cut of $\log z$ is chosen to be the positive real axis and $\mathrm{C}$ is the contour straddling the branch cut of $\log z$ starting from $a$ and ending at $a$ itself, as depicted in Figure-2, and the $z_k$'s are the poles of $g(z)$ enclosed by $\mathrm{C}$, with no pole of $g(z)$ lying along $\mathrm{C}$.
\end{lemma}
\begin{proof} 
	Given the contour integral $\int_{\mathrm{C}} g(z) \log z \, \mathrm{d}z$, deform the contour $\mathrm{C}$ into the branch cut by means of the contour $\mathrm{C}'$ in Figure-2, eventually taking the limit $\epsilon\rightarrow 0$.
\end{proof}

We will find that the first term of equation (\ref{contourep1}) yields the result of naive term by term integration; and the second term, the contribution coming from the poles, is the origin of the missing terms. In general, the integral representation of the incomplete Stieltjes transform relevant to us is of the form similar to equation (\ref{contourep1}). In particular, the representation will have a contour integral term which accounts for the term by term integration, and a term arising from contributions of the poles or branch points of the function $g(z)$, which are the missing terms. Here and in the following section, we will consider missing terms coming from pole singularities; later we will give an example of missing terms arising from a branch point. 

\begin{theorem} Let the complex extension, $f(z)$, of the function $f(x)$ be entire, then
	\begin{equation}\label{stel1}
	\int_0^a \frac{f(x)}{\omega+x}\, \mathrm{d}x =  \sum_{j=0}^{\infty} (-1)^j \omega^j \;\bbint{0}{a} \frac{f(x)}{x^{j+1}}\, \mathrm{d}x -f(-\omega) \ln\omega , \;\;\; 0< \omega<a .
	\end{equation}
\end{theorem}
\begin{proof}
	With $f(z)$ entire, the complex valued function $f(z)(\omega+z)^{-1}$ has a simple pole at $z=-\omega$. Then equation (\ref{contourep1}) leads to the following contour integral representation of the incomplete Stieltjes transform,
	\begin{equation}\label{integral1}
	\int_0^a \frac{f(x)}{\omega+x} \mathrm{d}x=-f(-\omega) \left(\ln\omega+ i\pi\right) + \frac{1}{2\pi i} \int_{\mathrm{C}} \frac{f(z) \log z}{\omega + z} \mathrm{d}z
	\end{equation}
	For fixed $\omega$ and $z\neq -\omega$, we have the expansion
	\begin{equation}\label{expansio}
	\frac{1}{\omega+z}= \sum_{j=0}^{n-1} \frac{(-1)^j \omega^j}{z^{j+1}} + \frac{(-1)^n \omega^n}{z^n (\omega+z)} , \;\;\; n=1, 2, \dots .
	\end{equation}
	Substituting this expansion back into the integral in equation (\ref{integral1}) yields
	\begin{eqnarray}
	\int_0^a \frac{f(x)}{\omega+x} \mathrm{d}x&=&-f(-\omega) \left(\ln\omega+ i\pi\right) +\sum_{j=0}^{n-1}(-1)^j \omega^j \frac{1}{2\pi i} \int_{\mathrm{C}} \frac{f(z) \log z}{z^{j+1}} \mathrm{d}z \nonumber \\
	&&\hspace{12mm} + (-1)^n \omega^n \frac{1}{2\pi i} \int_{\mathrm{C}} \frac{f(z) \log z}{z^n(\omega + z)} \mathrm{d}z\label{cow}
	\end{eqnarray}
 
 We now make the trivial replacement $\log z = \log z - i\pi + i\pi$ in the second term of equation (\ref{cow}),
 \begin{eqnarray}
 \int_0^a \frac{f(x)}{\omega+x} \mathrm{d}x&=&-f(-\omega) \left(\ln\omega+ i\pi\right) +\sum_{j=0}^{n-1}(-1)^j \omega^j \frac{1}{2\pi i} \int_{\mathrm{C}} \frac{f(z) }{z^{j+1}} (\log z - i\pi) \mathrm{d}z \nonumber \\
 && \hspace{-10mm}+ i \pi \frac{1}{2\pi i} \int_{\mathrm{C}} f(z)\left[ \sum_{j=0}^{n-1}(-1)^j    \frac{ \omega^j}{z^{j+1}} \right] \mathrm{d}z  + (-1)^n \omega^n \frac{1}{2\pi i} \int_{\mathrm{C}} \frac{f(z) \log z}{z^n(\omega + z)} \mathrm{d}z .
 \end{eqnarray}
 We recognize that the contour integral in the second term is just the finite part integral of the divergent integral $\int_0^a f(x) x^{-j-1}\mathrm{d}x$. We replace the summation in the third term with equation (\ref{expansio}). A residue term will emerge in the process which will cancel out the imaginary term in the first term. We arrive at the expression
 \begin{equation}
 \int_0^a \frac{f(x)}{\omega+x} \mathrm{d}x=-f(-\omega) \ln \omega  +\sum_{j=0}^{n-1}(-1)^j \omega^j \bbint{0}{a}  \frac{f(x)}{x^{j+1}} \mathrm{d}x + R_n(\omega) 
 \end{equation}
 where the remainder term is given by
 \begin{equation}
 R_n(\omega)=(-1)^n \omega^n \frac{1}{2\pi i} \int_{\mathrm{C}} \frac{f(z) }{z^n(\omega + z)} (\log z-i\pi) \mathrm{d}z
 \end{equation}

Despite appearances, the remainder is not a finite part integral because of the presence of the pole enclosed by the contour.	We estimate the remainder by deforming the contour into a circle centered at the origin with radius $a$. Then we have the bound
\begin{eqnarray}
\left|R_n(\omega)\right|&\leq& \frac{\omega^n}{a^n} \int_0^{2\pi}  \frac{|f(a \mathrm{e}^{i\theta})|}{|\omega + a \mathrm{e}^{i\theta}|}\left(\ln a + |\theta-\pi|\right) a \mathrm{d}\theta \nonumber\\
&\leq& \left(\frac{\omega}{a}\right)^n \frac{a  M(f,a) (\ln a + \pi^2)}{|a-\omega|} \label{bound}
\end{eqnarray}
in which $M(f,a)$ is the maximum modulo of $f(z)$ along the contour of integration. When $\omega<a$, the remainder term vanishes in the limit as $n\rightarrow\infty$. Then equation (\ref{stel1}) follows.
\end{proof}

\subsubsection{Example} Let us reproduce a known series expansion of the exponential integral function, $E_1(z)$, which assumes the integral representation \cite[p150,\#6.2.2]{nist}
\begin{equation}\label{expi}
E_1(z)= \mathrm{e}^{-z} \int_0^{\infty} \frac{\mathrm{e}^{-x}}{z+x}\, \mathrm{d}x,\;\;\; |\mathrm{Arg}z|<\pi .
\end{equation}
By inspection, the integral involved in the representation is a Stieltjes transform of the function $f(x)=\mathrm{e}^{-x}$.  We restrict ourselves to $z=\omega>0$ and just effect an analytic continuation in the complex plane if it is desired. Since the complex extension, $\mathrm{e}^{-z}$, of the real valued function $\mathrm{e}^{-x}$ is entire, it admits an expansion given by equation (\ref{stel1}), which we now derive. 

To proceed we cut the upper limit of integration to some finite $a>0$ and then eventually take the limit $a\rightarrow\infty$. The relevant finite part integrals to compute are $\bbint{0}{a}\mathrm{e}^{-x} x^{-j-1}\mathrm{d}t$. For some $\epsilon>0$, with $0<\epsilon<a$, we obtain, after term by term integration through the series expansion of $\mathrm{e}^{-x}$, the integral
\begin{eqnarray}
\int_{\epsilon}^{a} \frac{\mathrm{e}^{-x}}{x^{j+1}} \mathrm{d}x &=& -\sum_{k=0}^{n-1} \frac{(-1)^k}{k! (j-k)} \left(\frac{1}{a^{j-k}}-\frac{1}{\epsilon^{j-k}}\right) + \frac{(-1)^j}{j!}\left(\ln a -\ln \epsilon\right)\nonumber \\
&& + \sum_{k=j+1}^{\infty} \frac{(-1)^k}{k!(k-j)} \left(a^{k-j}-\epsilon^{k-j}\right) .\label{xxx}
\end{eqnarray}
The finite part of the divergent integral $\int_0^a \mathrm{e}^{-x} x^{-j-1} \mathrm{d}x$ is obtained by dropping all terms that diverge as $\epsilon\rightarrow 0$ and keeping only the terms with a well defined limit in the said limit in equation (\ref{xxx}). Then we obtain the finite part integral
\begin{eqnarray}
\bbint{0}{a} \frac{\mathrm{e}^{-x}}{x^{j+1}}\, \mathrm{d}x &=& - \sum_{k=0}^{j-1} \frac{(-1)^k}{k! (j-k)}\frac{1}{a^{j-k}}\nonumber \\&&\hspace{10mm} + \frac{(-1)^j}{j!} \ln a + \frac{ (-1)^{j+1} a}{(j+1)!} \, _2F_2(1,1;2,j+2;-a), \label{fp1}
\end{eqnarray}
where the contributing infinite series in equation (\ref{xxx}) has been summed in terms of the hypergeometric function $_2F_2(a_1,a_2;b_1,b_2;z)$.

We now take the limit as $a$ approaches infinity. Only the last two terms contribute in the limit.  The hypergeometric function involved has a double pole, so that its relevant asymptotic expansion needed to compute the limit is given by \cite{wolfram1}
\begin{eqnarray}
_2F_2\left(a_1,a_1;b_1,b_2;z\right)& = & \frac{\Gamma(b_1) \Gamma(b_2)}{\Gamma(a_1)^2} \mathrm{e}^{z} \left(1+O(z^{-1})\right) z^{2 a_1-b_1-b_2} \nonumber \\
&& \hspace{-16mm}+ \frac{\Gamma(b_1) \Gamma(b_2)}{\Gamma(a_1) \Gamma(b_1-a_2) \Gamma(b_2-a_1)} \left[\log(-z) \left(1+O(z^{-1})\right) \right. \nonumber \\
&&\hspace{-16mm} \left. - \left(\psi(b_1-a_1) + \psi(b_2-a_1) + \psi(a_1) + 2\gamma\right)\left(1+ O(z^{-1})\right)\right],\;\; |z|\rightarrow\infty ,
\end{eqnarray}
where $\psi(z)$ is the logarithmic derivative of the gamma function and $\gamma$ is the Euler constant, with $\psi(1)=-\gamma$. Only the second term contributes in the limit.

Substituting the appropriate values of the parameters of the hypergeometric function and implementing the limit, we obtain the desired finite part integral
\begin{equation}
\bbint{0}{\infty} \frac{\mathrm{e}^{-x}}{x^{j+1}}\, \mathrm{d}x = \frac{(-1)^j}{j!} \psi(j+1), \;\;\; j=0, 1, 2, \dots .
\end{equation}
This result cannot be obtained by analytic continuation. The expansion of the Stieltjes transform of $\mathrm{e}^{-t}$ given by equation (\ref{stel1}) can now obtained and substituted back into equation (\ref{expi}) to yield the following known expansion of the exponential integral function \cite[p151,\# 6.6.3]{nist}
\begin{equation}
E_1(\omega)= - \ln \omega + \mathrm{e}^{-\omega} \sum_{j=0}^{\infty} \frac{\psi(j+1)}{j!}  \omega^j .
\end{equation}  
Absolute convergence of the infinite series is guaranteed by the asymptotic behavior $\psi(j+1)\sim \ln j$ for arbitrarily large $j$. This must be as Theorem-3 assures convergence for all $\omega<a$, which is satisfied for all $\omega$ as $a$ becomes arbitrarily large.

\subsection{Case 2:} $\nu\neq 0$ 
\begin{lemma}\label{lemmabranch}
Let $g(x)$ be integrable in the interval $[0,a]$ and that its complex extension $g(z)$ is analytic in a region containing the interval $[0,a]$ in its interior. Let $z^{-\nu}$, for $0<\nu<1$, be such that its branch cut coincides with the real positive axis. Then
	\begin{equation}
	\int_0^a x^{-\nu} g(x) \mathrm{d}x = \frac{1}{\left(\mathrm{e}^{-2\pi \nu i}-1\right)} \int_{\mathrm{C}} z^{-\nu} g(z)\, \mathrm{d}z - \frac{2\pi i}{\left(\mathrm{e}^{-2\pi \nu i}- 1\right)} \sum_k \mathrm{Res}\left[z^{-\nu} g(z)\right]_{z=z_k},
	\end{equation}
	where $\mathrm{C}$ is the contour straddling the branch cut of $z^{-\nu}$ starting from $a$ and ending at $a$ itself, and the $z_k$'s are the poles enclosed by $\mathrm{C}$, with no pole of $g(z)$ lying on the contour $\mathrm{C}$.
\end{lemma}
\begin{proof}
	Given the contour integral $\int_{\mathrm{C}} g(z) z^{-\nu} \, \mathrm{d}z$, deform the contour $\mathrm{C}$ into the branch cut by means of the contour $\mathrm{C}'$ in Figure-2, again eventually taking the limit $\epsilon\rightarrow 0$.
\end{proof}

\begin{theorem} Let the complex extension, $f(z)$, of $f(x)$ be entire. Then
	\begin{equation} \label{stel2}
	\int_0^a \frac{x^{-\nu} f(x)}{\omega + x} \mathrm{d}x = \sum_{j=0}^{\infty} (-1)^j \omega^j \bbint{0}{a}\frac{f(x)}{x^{j+\nu+1}}\, \mathrm{d}x + \frac{\pi f(-\omega)}{\omega^{\nu} \sin\pi\nu} ,\;\;\;0< \omega<a, \;\;\; 0<\nu<1 .
	\end{equation}
\end{theorem}
\begin{proof}
From Lemma-\ref{lemmabranch} we have the following contour integral representation of the incomplete Stieltjes transform
\begin{equation}\label{second}
\int_0^a \frac{x^{-\nu} f(x)}{\omega + x} \mathrm{d}x = \frac{\pi f(-\omega)}{\omega^{\nu} \sin\pi \nu} + \frac{1}{\left(\mathrm{e}^{-2\pi\nu i}-1\right)} \int_{\mathrm{C}} \frac{f(z)}{z^{\nu} (\omega+z)} \mathrm{d}z
\end{equation}
We insert the same expansion for $(\omega+z)^{-1}$ in the integral and use the contour integral representation of the finite part integral to obtain 
\begin{equation}
\int_0^a \frac{x^{-\nu} f(x)}{\omega + x} \mathrm{d}x = \frac{\pi f(-\omega)}{\omega^{\nu} \sin\pi \nu} +\sum_{j=0}^{n-1} (-1)^j \bbint{0}{a} \frac{f(x)}{x^{j+1+\nu}} \mathrm{d}x + R_n(\omega)
\end{equation}
where the remainder term is given by
\begin{equation}
R_n(\omega)=\frac{(-1)^n \omega^n}{(\mathrm{e}^{-2\pi\nu i}-1)} \int_{\mathrm{C}} \frac{f(z)}{z^{n+\nu} (\omega + z)}\mathrm{d}z .
\end{equation}

Again deforming the contour to be the circular path with radius $a$ and centered at the origin, we can establish the bound
\begin{equation}\label{bound2}
|R_n(\omega)|\leq \left(\frac{\omega}{a}\right)^n \frac{a M(f,a) \pi}{a^{\nu} |\omega-\nu|  |\sin\pi\nu|} .
\end{equation}
For $\omega<a$ the remainder vanishes as $n$ approaches infinity. This establishes the expansion given by equation (\ref{stel2}).

\end{proof}

\subsubsection{Example} This time let us reproduce a known series expansion of the incomplete gamma function, $\Gamma(a,z)$, which assumes the integral representation \cite[p177,\# 8.6.4]{nist}
\begin{equation}
\Gamma(\nu,z)= \frac{z^{\nu} \mathrm{e}^{-z}}{\Gamma(1-\nu)} \int_0^{\infty} \frac{x^{-\nu} \mathrm{e}^{-x}}{z+x}\mathrm{d}x, \;\; |\mathrm{Arg}z\,|<\pi, \;\; \mathrm{Re}\,\nu<1 .
\end{equation}
We recognize that the integral in the representation is the Stieltjes transform of the function $x^{-\nu} \mathrm{e}^{-x}$, which falls under equation (\ref{stel2}) with $f(z)=\mathrm{e}^{-z}$. In this example, we again make the restriction $z=\omega>0$ and further restrict ourselves to the case $0<\mathrm{Re}\, \nu<1$. Then the expansion given by equation \ref{second} applies.

We proceed as above in obtaining the required finite part integrals. For some $\epsilon$, with $0<\epsilon<a$, we have the following integral 
\begin{eqnarray}
\int_{\epsilon}^{a} \frac{\mathrm{e}^{-x}}{x^{j+1+\nu}}\mathrm{d}x &=& \sum_{k=0}^j \frac{(-1)^k}{k! (k-j-\nu)} \left(\frac{1}{a^{j+\nu-k}}-\frac{1}{\epsilon^{j+\nu-k}}\right) \nonumber \\
&& + \sum_{k=j+1}^{\infty} \frac{(-1)^k}{k! (k-j-\nu)} \left(a^{k-j-\nu}-\epsilon^{k-j-\nu}\right)
\end{eqnarray}
The finite part of the integral is again just the terms that have a finite limit as $\epsilon\rightarrow 0$. Collecting all terms with finite limits, the finite part integral, for finite $a$, is
\begin{eqnarray}
\bbint{0}{a}\frac{\mathrm{e}^{-x}}{x^{j+1+\nu}} \mathrm{d}x &=&\sum_{k=0}^j \frac{(-1)^k}{k! (k-j-\nu)} \frac{1}{a^{j+\nu-k}}\nonumber \\
&& + (-1)^{j+1} \frac{a^{1-\nu}}{(1-\nu) (j+1)!}\, _2F_2\left(1,1-\nu;2+j,2-\nu;-a\right),
\end{eqnarray}
where the sum in the second term has been expressed as a hypergeometric function.

Finally we take the limit as $a\rightarrow\infty$. The first term vanishes in the limit. The hypergeometric function has two simple poles, so its relevant asymptotic expansion in the computation of the limit is given by \cite{wolfram2}
\begin{eqnarray}
\, _2F_2(a_1,a_2;b_1,b_2;z)& = & \frac{\Gamma(b_1) \Gamma(b_2) \Gamma(a_2-a_1)}{\Gamma(a_2) \Gamma(b_1-a_1) \Gamma(b_2-a_1)} (-z)^{-a_1}\left(1+O(z^{-1})\right) \nonumber \\
&& +\frac{\Gamma(b_1) \Gamma(b_2) \Gamma(a_1-a_2)}{\Gamma(a_1) \Gamma(b_1-a_2) \Gamma(b_2-a_2)} (-z)^{-a_2}\left(1+O(z^{-1})\right) \nonumber \\
&& + \frac{\Gamma(b_1) \Gamma(b_2)}{\Gamma(a_1)\Gamma(a_2)} \mathrm{e}^{z} z^{a_1+a_2-b_1-b_2} \left(1+O(z^{-1})\right),\;\;\;\; |z|\rightarrow \infty
\end{eqnarray}
for $a_1\neq a_2$.  Only the second term contributes in the limit under the condition that $\mathrm{Re}\,\nu>0$.

Substituting the appropriate values of the parameters of the hypergeometric function and implementing the limit gives the desired finite part integral,
\begin{equation}
\bbint{0}{\infty} \frac{\mathrm{e}^{-x}}{x^{j+1+\nu}}\mathrm{d}x = (-1)^{j+1} \frac{\Gamma(1-\nu)\Gamma(\nu)}{\Gamma(j+\nu+1)} .
\end{equation}
This result can be obtained also by analytic continuation. We substitute this back into equation (\ref{stel2}) to obtain the expansion of the Stieltjes transform and use the well-known identity
$\Gamma(z) \Gamma(1-z)=\pi \csc\pi z$
to finally obtain the following known expansion for the incomplete gamma function \cite[p178,\# 8.7.3]{nist}
\begin{equation}
\Gamma(a,z) = \Gamma(a) \left[1 - \mathrm{e}^{-z} \sum_{j=0}^{\infty} \frac{z^{j+\nu}}{\Gamma(j+\nu+1)}\right] .
\end{equation} 
Notice that the series is clearly convergent everywhere. This is expected as in the limit $a\rightarrow\infty$ the upper limit $a$ becomes arbitrarily large so that the condition $\omega<a$ is satisfied for any finite $\omega$.

\subsection{Remark}
\begin{figure}
	\includegraphics[scale=0.4]{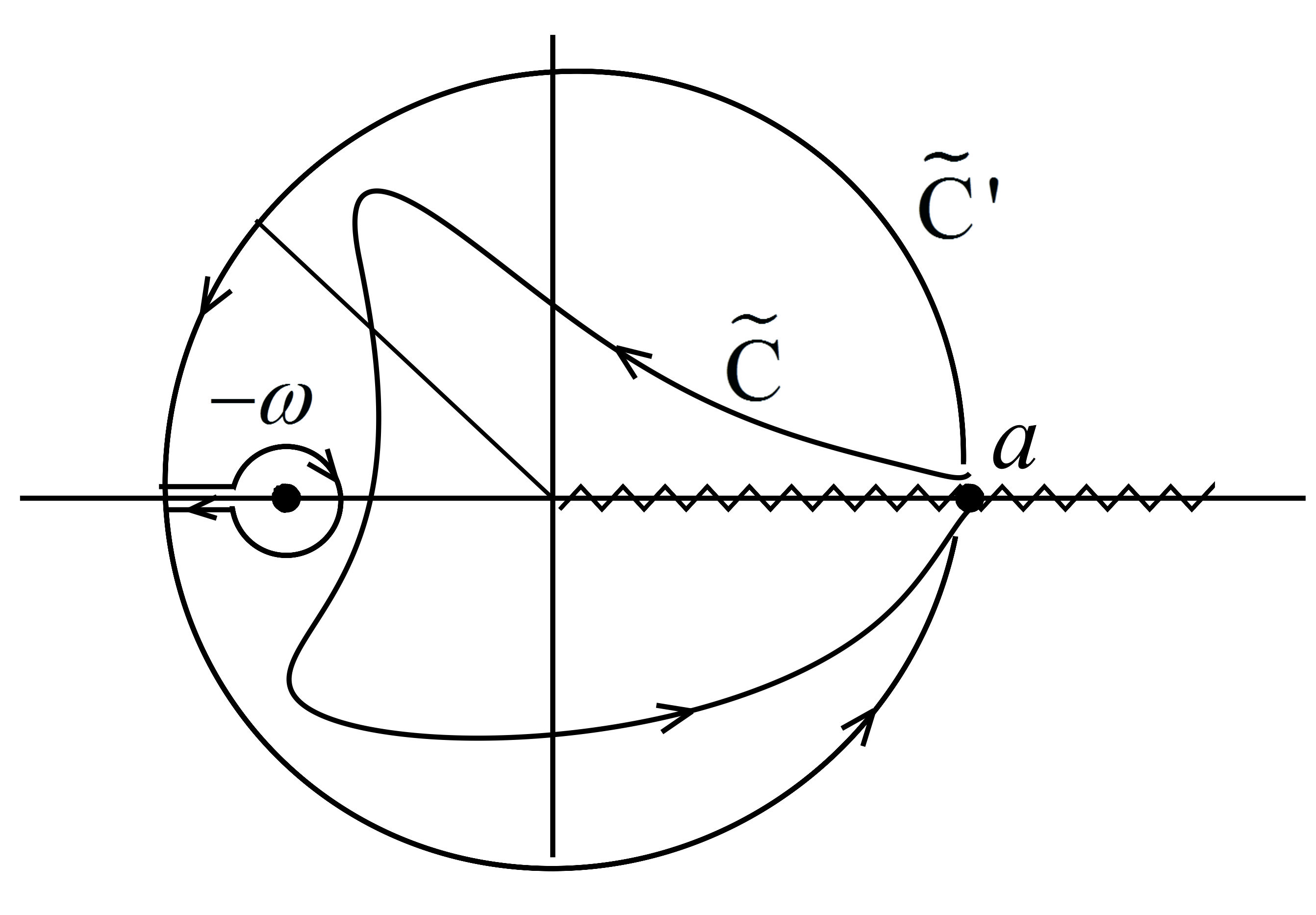}
	\caption{Missing out the pole term.}
	\label{missing}
\end{figure}

What happens if we choose a contour $\widetilde{\mathrm{C}}$ that excludes the lone pole at $z=-\omega$? It appears that under this condition the pole contribution will drop out and only the term by term integrated terms will remain in the expansion, apparently leading to a contradiction with our results above. The apparent contradiction is resolved by recognizing that the equality in equations (\ref{stel1}) and (\ref{stel2}) only hold under the condition that the remainder vanishes as the number of terms in the partial sum increases indefinitely. When the pole is excluded, the remainder will not vanish in the limit; and, when the remainder is appropriately dealt with, we will arrive at the same result.

We can demonstrate  this with the later case of $\nu\neq 0$. Using Lemma-\ref{lemmabranch} and following the same steps above, we obtain the finite expansion of the integral as follows
\begin{equation}
\int_0^a \frac{x^{-\nu} f(x)}{\omega + x} \mathrm{d}x = \sum_{j=0}^{n-1} (-1)^j \bbint{0}{a} \frac{f(x)}{x^{j+1+\nu}} \mathrm{d}x + \widetilde{R}_n(\omega),
\end{equation}
where the remainder term is given by
\begin{equation}
\widetilde{R}_n(\omega)=\frac{(-1)^n \omega^n}{(\mathrm{e}^{-2\pi\nu i}-1)} \int_{\widetilde{\mathrm{C}}} \frac{f(z)}{z^{n+\nu} (\omega + z)}\mathrm{d}z .
\end{equation}

Contradiction with the result (\ref{stel2}) arises if the remainder $\widetilde{R}_n(\omega)$ vanishes as $n\rightarrow\infty$.  But it does not vanish. We can see that by deforming the contour $\widetilde{C}$ into the circular contour $\widetilde{C}'$ depicted in Figure \ref{missing}. Along this contour, the remainder term assumes the form
\begin{equation}
\widetilde{R}_n(\omega)= \frac{(-1)^n \omega^n}{(\mathrm{e}^{-2\pi\nu i}-1)} \int_{\widetilde{\mathrm{C}_a}} \frac{f(z)}{z^{n+\nu} (\omega + z)}\mathrm{d}z - 2\pi i \mathrm{Res}\left[\frac{(-1)^n \omega^n}{(\mathrm{e}^{-2\pi\nu i}-1)} \frac{f(z)}{z^{n+\nu} (\omega + z)}\right]_{z=-\omega} .
\end{equation}
With the bound given by equation (\ref{bound2}), the first term vanishes in the limit as $n\rightarrow\infty$. On the other hand, evaluating the residue on the second term yields a term independent of $n$ and in fact equal to the contribution of the simple pole at $z=-\omega$ in equation (\ref{stel2}). That is the remainder term carries the pole contribution excluded by the contour $\widetilde{C}$, so that the result given by equation (\ref{stel2}) is recovered in the limit as $n\rightarrow\infty$. This example emphasizes the need to consider the remainder term in the application of our approach.

\section{Term by term integration at the origin for functions with holomorphic extensions}\label{origin2}
We now consider the case in which the complex extension function $f(z)$ is holomorphic in the complex plane. Let $z_1$, $z_2$, $z_3$, $\dots$ be the poles of $f(z)$, and $\zeta_0$ the greatest lower bound of the modulus of the poles, i.e. $\zeta_0=\inf\{|z_1|, |z_2|, \dots\}$. Then we have the following results. 

\begin{figure}
	\includegraphics[scale=0.4]{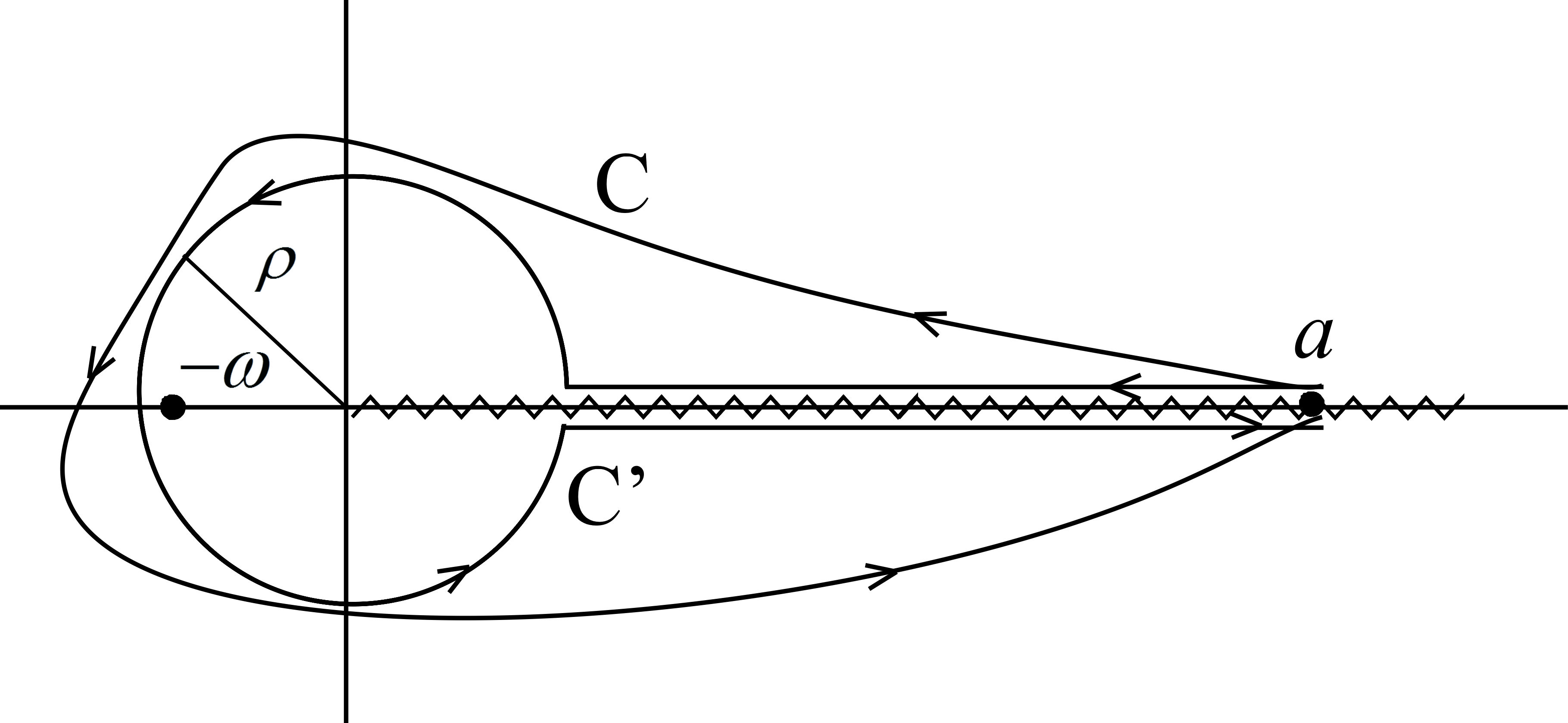}
	\caption{The contour in the holomorphic case.}
	\label{holomorphic}
\end{figure}

\begin{theorem}\label{lees}
	If $\zeta_0 \neq 0$, then for every positive $\omega<a, \zeta_0$ the following expansion holds
	\begin{equation}
		\int_0^a \frac{ f(x)}{(\omega+x)}\mathrm{d}x = \sum_{j=0}^{\infty} (-1)^j \omega^j \bbint{0}{a} \frac{f(x)}{x^{j+1}} \mathrm{d}x - f(-\omega) \ln\omega . \label{result4}
	\end{equation}
\end{theorem}	
\begin{proof}
	Let the contour $\mathrm{C}$ be sufficiently large to accommodate the simple pole at $z=-\omega$ but not any of the poles of $f(z)$ (see Figure-\ref{holomorphic}). Then by the representation Lemma-\ref{lemma0} and by following the same steps we have used in the entire case, we arrive at the expansion
	\begin{eqnarray}
	\int_0^a \frac{f(x)}{\omega+x}\mathrm{d}x = - f(-\omega) \ln \omega + \sum_{j=0}^{n-1} (-1)^j \omega^j \bbint{0}{a} \frac{f(x)}{x^{j+1}} \mathrm{d}x + R_n(\omega) ,
	\end{eqnarray}
	where the remainder term is given by
	\begin{equation}\label{remainder2}
	R_n(\omega) = (-1)^n \omega^n \frac{1}{2\pi i} \int_{\mathrm{C}} \frac{f(z) (\log z -i\pi)}{z^n (\omega+z)} \mathrm{d}z .
	\end{equation}
	
	Under the condition that $\omega<a, \zeta_0$, there always exists some $\rho$ such that $\omega<\rho<a, \zeta_0$. We then deform the contour of integration into the contour $\mathrm{C}'$ shown in Figure-\ref{holomorphic}. Along this contour, the remainder assumes the form
	\begin{equation}\label{bla}
	R_n(\omega)= (-1)^n \omega^n \frac{1}{2\pi i} \int_{\mathrm{C}_{\rho}} \frac{f(z) (\log z -i\pi)}{z^n (\omega+z)} \mathrm{d}z + (-1)\omega^n \int_{\rho}^{a} \frac{f(x)}{x^n (\omega+x)}\mathrm{d}x .
	\end{equation}
	We can use the bound given by inequality (\ref{bound}) for the first term of equation (\ref{bla}) to obtain the bound for the remainder term
	\begin{equation}
	\left|R_n(\omega)\right| \leq \left(\frac{\omega}{\rho}\right)^n  \left[ \frac{M(f,\rho) (\ln\rho + \pi^2)}{|\rho-\omega|} + \int_{\rho}^{a} \frac{\left|f(x)\right|}{\omega + x} \mathrm{d}x\right].\label{bound3}
	\end{equation}
	Since $\omega<\rho$ the bound vanishes as $n\rightarrow \infty$. Then equation (\ref{result4}) follows.
\end{proof}

Using similar steps, we can establish the following result.

\begin{theorem}
	If $\zeta_0\neq 0$, then for every positive $\omega<a,\zeta_0$ the following expansion holds
	\begin{equation}
	\int_0^a \frac{x^{-\nu} f(x)}{(\omega+x)}\mathrm{d}x = \sum_{j=0}^{\infty} (-1)^j \omega^j \bbint{0}{a} \frac{f(x)}{x^{j+\nu+1}} \mathrm{d}x + \frac{\pi f(-\omega)}{\omega^{\nu} \sin\pi\nu} , \;\;\; 0<\nu<1 .
	\end{equation}
\end{theorem}

\section{Term by term integration at infinity}\label{infinity}
Here we apply our method to a problem that has been treated with the distributional approach \cite{wong,distribution,mcwong}. We will only demonstrate how the method can be applied when expansions are made at infinity. A detailed exposition of the relationship between our approach and the distributional approach to asymptotics will be considered elsewhere. Let us consider the integral
\begin{equation}\label{exampleinf}
I(\omega)= \int_0^{\infty} \frac{1}{(1+x)^{\nu} (\omega + x)} \mathrm{d}x, \;\;\; 0<\nu<1, \;\;\; 0<\omega .
\end{equation}
We wish to obtain the asymptotic expansion for $\omega\rightarrow\infty$. The problem discussed in the introduction is a special case of this integral for $\nu=1/3$. We leave the remainder analysis to the reader.

The first step is to obtain a contour integral representation of the integral given by equation \ref{exampleinf}. The complex extension function $(1+z)^{-\nu} (\omega+z)^{-1}$ has a branch point at $z=-1$. We choose the branch of $(1+z)^{-\nu}$ such that its branch cut is the line $[-1,\infty)$. Furthermore, we choose the contour given in Figure-\ref{atinfty}. Since $\omega$ is arbitrarily large, the pole at $z=-\omega$ is outside this contour. 

\begin{figure}
	\includegraphics[scale=0.4]{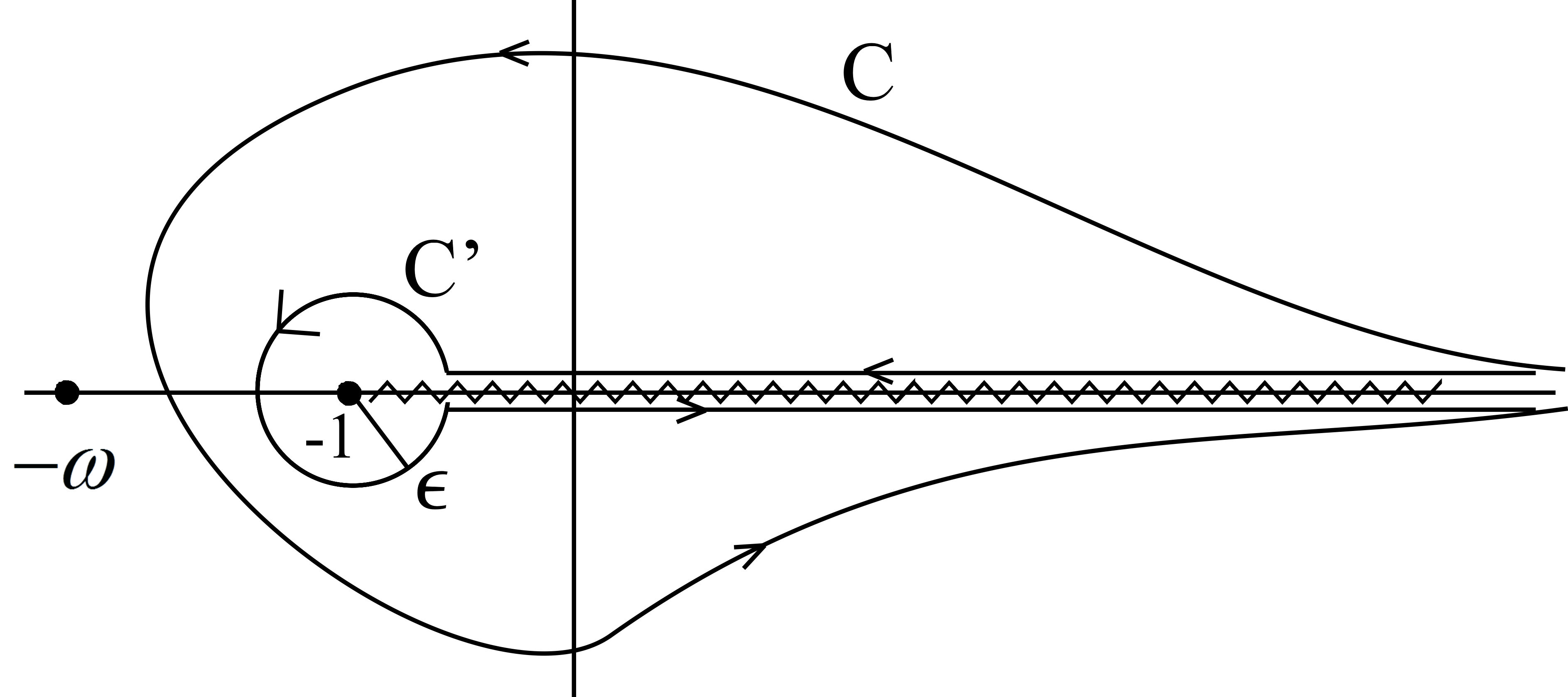}
	\caption{The contour of integration for a term by term integration at infinity.}
	\label{atinfty}
\end{figure}

By deforming the contour into the branch cut of $(1+z)^{-\nu}$, we obtain the contour integral representation of the given integral,
\begin{eqnarray}
\int_0^{\infty} \frac{1}{(1+x)^{\nu}(\omega+x)}\mathrm{d}x &=& -\int_0^1 \frac{1}{(1-x)^{\nu} (\omega-x)}\mathrm{d}x \nonumber \\
&&\hspace{12mm} + \frac{1}{\left(\mathrm{e}^{-2\pi \nu i }-1\right)} \int_{\mathrm{C}} \frac{1}{(1+z)^{\nu} (\omega+z)} \mathrm{d}z\label{x4}
\end{eqnarray}
The first integral exists because $\omega$ is arbitrarily large, i.e. $\omega>>1$. This time we use the expansion
\begin{equation}
\frac{1}{(1+z)^{\nu}}=\sum_{s=0}^{\infty} \binom{-\nu}{s} \frac{1}{z^{s+\nu}},
\end{equation}
valid for $|z|>1$.

We substitute this expansion back into the second term of equation (\ref{x4}) and make the necessary identification for the finite part integral to obtain the expansion
\begin{eqnarray}
\int_0^{\infty} \frac{1}{(1+x)^{\nu}(\omega+x)}\mathrm{d}x &=& -\int_0^1 \frac{1}{(1-x)^{\nu} (\omega-x)}\mathrm{d}x \nonumber \\
&&\hspace{12mm} + \sum_{s=0}^{\infty} \binom{-\nu}{s} \bbint{0}{\infty} \frac{1}{(\omega+x) x^{s+\nu}} \mathrm{d}x .\label{quay}
\end{eqnarray}
Observe that the second term is just the term by term integration of the expansion, with the integration interpreted as a finite part integral. The first term represents the contribution coming from the branch point at $z=-1$. The expansion in the second term does not have an accompanying pole term because the pole at $z=-\omega$ is outside of the contour of integration.

Let us calculate the indicated finite part integral. For $0<\epsilon<a<\infty$, we have the convergent integral
\begin{eqnarray}
\int_{\epsilon}^a \frac{1}{(\omega+x) x^{s+\nu}} &=& \int_{\epsilon}^{a} \frac{1}{x^{s+\nu}} \left[\sum_{j=0}^{s-1} (-1)^j\frac{x^j}{\omega^{j+1}} + \frac{(-1)^s}{\omega^s} \frac{x^s}{(\omega+x)}\right]  \mathrm{d}x \nonumber\\
&=& - \sum_{j=0}^{s-1} \frac{1}{(s+\nu-j-1) \omega^{j+1}} \left[\frac{1}{a^{s+\nu-j-1}}-\frac{1}{\epsilon^{s+\nu-j-1}}\right] \nonumber \\
&& + \frac{(-1)^s}{\omega^s} \int_{\epsilon}^{a} \frac{1}{x^{\nu} (\omega + x)}\mathrm{d}x 
\end{eqnarray}
The finite part for $a<\infty$ is obtained by dropping the terms that diverge as $\epsilon\rightarrow 0$, and keeping only those terms with a well defined limit in the said limit. Finally taking the limit as $a\rightarrow\infty$ yields the desired finite part integral
\begin{eqnarray}
\bbint{0}{\infty} \frac{1}{(\omega+x) x^{s+\nu}} \mathrm{d}x = \frac{(-1)^s}{\omega^s}\int_0^{\infty} \frac{1}{x^{\nu} (\omega+x)}\mathrm{d}x = \frac{(-1)^s \pi}{\omega^{s+\nu} \sin\pi\nu} .
\end{eqnarray}
Observe that the finite part is equal to the analytic value of the divergent integral.

Finally the first integral can be expanded by expanding $(\omega-x)^{-1}$ at $x=0$,  and then integrating the infinite series term by term. Here no special treatment is needed for the term by term integration because no divergent integral is involved, and the infinite series for the expansion is absolutely convergent in the entire interval of integration for $\omega>>1$. By substitution of the expansion of the first term and the finite part integrals back in equation \ref{quay}, we obtain the desired expansion,
\begin{eqnarray}
\int_0^{\infty} \frac{1}{(1+x)^{\nu} (\omega + x)} \mathrm{d}x &=&  - \sum_{s=0}^{\infty} \frac{\Gamma(s+1)  \Gamma(1-\nu)}{\Gamma(s-\nu+2)} \frac{1}{\omega^{s+1}} \nonumber \\
&& \hspace{12mm} + \frac{\pi}{\sin\pi\nu} \sum_{s=0}^{\infty} (-1)^s \binom{-\nu}{s} \frac{1}{\omega^{s+\nu}}, \;\;\; \omega\rightarrow \infty . \label{infinite}
\end{eqnarray}
Both infinite series in equation \ref{infinite} converge for sufficiently large $\omega$, in particular for $\omega>1$; and the indicated equality is an exact equality not mere asymptotic equality. We recover from equation (\ref{infinite}) the result for $\nu=1/3$ given earlier in equation (\ref{correct}). 

\section{Conclusion}

In this paper we tackled the problem of evaluating the Stieltjes transform by means of expanding the integrand at infinity and then evaluating the series term by term. The interchange of order of integration and summation led to an infinite series whose terms are divergent integrals. We have seen that committing ourselves to a particular interpretation of the divergent integrals and naively assigning them values in accordance with the chosen interpretation will generally lead to missing terms. We have seen that the resolution to this problem lies in the precise relationship between the assigned values of the divergent integrals and the Stieltjes integral. Here we have chosen to assign the divergent integrals the value equal to their finite parts. The relationship between the finite part integral and the Stieltjes transform is established by lifting the integration in the complex plane, in particular, by representing them as complex contour integrals with common contours of integration. The missing terms are recovered in the process and found to arise from the poles and branch points of the integrand in the complex plane. Moreover, the result is not a mere asymptotic equality but an exact analytic equality, demonstrating that divergent integrals, specifically their finite parts, can in fact be used to evaluate convergent integrals.

In hindsight, we have not only addressed the question of the origin of missing terms but in the process have uncovered the method of finite part integration as an alternative integration method of evaluating convergent integrals.  In general, some uniformity conditions must be satisfied in order for integration and summation to be interchanged. If the interchange leads to divergent integrals, then the necessary conditions are not satisfied and the interchange is not allowed. Finite part integration shows us how to proceed with the interchange when the interchange leads to a finite or infinite sum  of divergent integrals to obtain the exact analytic value of a given convergent integral. To perform finite part integration on a given convergent integral, one generally has to proceed as follows: Determine the divergent integrals that arise after expanding the integrand of the given integral and performing a term by term integration; obtain the finite parts of the divergent integrals; obtain the complex contour integral representation of the finite parts; represent the given integral as a complex contour integral using the same contour of integration as the finite part integrals; perform the expansion of the integrand under the contour integral of the given (convergent) integral; proceed with the term by term integration. If the expansion of the integrand has a fixed finite number of terms, finite part integration may proceed without complication.  
But a non-trivial consideration of the remainder term is necessary when the expansion of the integrand leads to an infinite series. Elsewhere we will develop further finite part integration, and explore its applications to asymptotics and to perturbative solutions to differential or integral equations involving infinite series whose terms are divergent integrals.


\end{document}


\title[Term by term integration]{The limit to infinity: Addendum to ``The problem of missing terms in term by term integration involving divergent integrals''}

\author{Eric A. Galapon}
\address{Theoretical Physics Group, National Institute of Physics, University of the Philippines, Diliman Quezon City, 1101 Philippines}
\email{eagalapon@up.edu.ph}
\date{\today}
\maketitle

\begin{abstract} In [E.A. Galapon {\it Proc. Roc. Soc. A } 473: 20160567 (2017)] we considered the problem of missing terms in term by term integration involving divergent integrals arising from the evaluation of the incomplete Stieltjes transform $\int_0^a f(x) (\omega+x)^{-1} \mathrm{d}x$, and the Stieltjes transform itself, $\int_0^{\infty} f(x) (\omega+x)^{-1} \mathrm{d}x$. The evaluation of the Stieltjes transform was assumed without proof from its incomplete version. Here we give the explicit proof for the limiting transition from finite $a$ to infinite $a$. 
\end{abstract}

\section{Introduction}
In \cite{galaponPRSA2017} we considered the problem of missing terms arising from evaluating the incomplete Stieltjes transform,
\begin{equation}
S_a(\omega)=\int_0^a  \frac{x^{-\nu}f(x)}{(\omega+x)}\,\mathrm{d}x, \;\;\; 0<\omega<a <\infty ,\;\; 0\leq \nu<1,
\end{equation} 
following from binomially expanding the kernel $(\omega+x)^{-1}$ about $\omega=0$ and then integrating the resulting infinite series term by term. The interchange  of summation and integration leads to the infinite series  $\sum_{j=0}^{\infty}(-1)^j \omega^j\int_0^a x^{-(j+\nu+1)} f(x)\,\mathrm{d}x$ whose terms are divergent integrals. Assigning values to the (divergent) integrals by, say, analytic continuation yields a result that completely misses out a group terms of the exact result. The problem of the missing terms was addressed in \cite{galaponPRSA2017} by lifting the integration in the complex plane. 

Under the sole assumption that the function $f(x)$ possesses an analytic extension, $f(z)$, in the interval $[0,a]$, the incomplete Stieltjes integral assumes the evaluation
	\begin{equation}\label{stel1}
	\int_0^a \frac{f(x)}{\omega+x}\, \mathrm{d}x =  \sum_{j=0}^{\infty} (-1)^j \omega^j \bbint{0}{a} \frac{f(x)}{x^{j+1}}\, \mathrm{d}x -f(-\omega) \ln\omega .
	\end{equation}
	\begin{equation} \label{stel2}
	\int_0^a \frac{x^{-\nu} f(x)}{\omega + x} \mathrm{d}x = \sum_{j=0}^{\infty} (-1)^j \omega^j \bbint{0}{a}\frac{f(x)}{x^{j+\nu+1}}\, \mathrm{d}x + \frac{\pi f(-\omega)}{\omega^{\nu} \sin\pi\nu} ,\;\;\; 0<\nu<1 ,
	\end{equation}	
for all $\omega<a, |z_0|$, where $z_0$ is the nearest singularity of $f(z)$ to the origin and $\bbint{0}{a} x^{-j-\nu-1} f(x)\,\mathrm{d}x$ is the finite part of the divergent integral $\int_0^a x^{-j-\nu-1} f(x)\,\mathrm{d}x$. 
For $\nu=0$, the finite parts are given by 
	\begin{equation}\label{fpi1}
	\bbint{0}{a} \frac{f(x)}{x^{n+1}} \mathrm{d}x = \lim_{\epsilon\rightarrow 0^+}\left[\int_{\epsilon}^a \frac{f(x)}{x^{n+1}} -\sum_{k=0}^{n-1}\frac{f^{(k)}(0)}{k! (n-k)}\frac{1}{\epsilon^{n-k}} + \frac{f^{(n)}(0)}{n!}\ln \epsilon\right] ,
	\end{equation}
for all $n=0, 1, 2, \dots$; and for $0<\nu<1$, the finite parts are
	\begin{equation}\label{fpi2}
	\bbint{0}{a} \frac{f(x)}{x^{m+\nu}} \mathrm{d}x = \lim_{\epsilon\rightarrow 0^+}\left[\int_{\epsilon}^{a}\frac{f(x)}{x^{m+\nu}}\mathrm{d}x - \sum_{j=0}^{m-1} \frac{f^{(j)}(0)}{j! (n+\nu-j-1)}\frac{1}{\epsilon^{n+\nu-j-1}}\right]
	\end{equation}
for all $m=1, 2, \dots$. The analyticity of $f(z)$ at the origin guarantees the existence of the above limits. 
		
The infinite series in equations \eqref{fpi1} and \eqref{fpi2} are the results of term by term integration with the divergent integrals assigned values equal to their finite parts, and the terms proportional to $f(-\omega)$ are the terms that would be missed out when the divergent integrals were merely replaced with their finite parts. 
	
Now when the upper limit of integration is infinity, we obtain the well-known Stieltjes transform which is defined by the limit 
\begin{equation}
\int_0^{\infty} \frac{f(x)}{\omega + x} \, \mathrm{d}x = \lim_{a\rightarrow\infty} \int_0^{a} \frac{f(x)}{\omega + x} \, \mathrm{d}x .
\end{equation}
It was tacitly claimed without proof in \cite{galaponPRSA2017} that equations \eqref{stel1} and \eqref{stel2} hold for $a=\infty$ which is equivalent to interchanging the limit and the summation. This claim was demonstrated by three examples. 
Of course it is not necessary that the result of a pair of sequence of operations is preserved by the interchange of the operations, and one can raise legitimately the possibility that the examples are mere coincidences. Here we give the explicit proof that indeed the interchange of the order of the summation and the limit is justified under the conditions stated there.

\begin{figure}
	\includegraphics[scale=0.4]{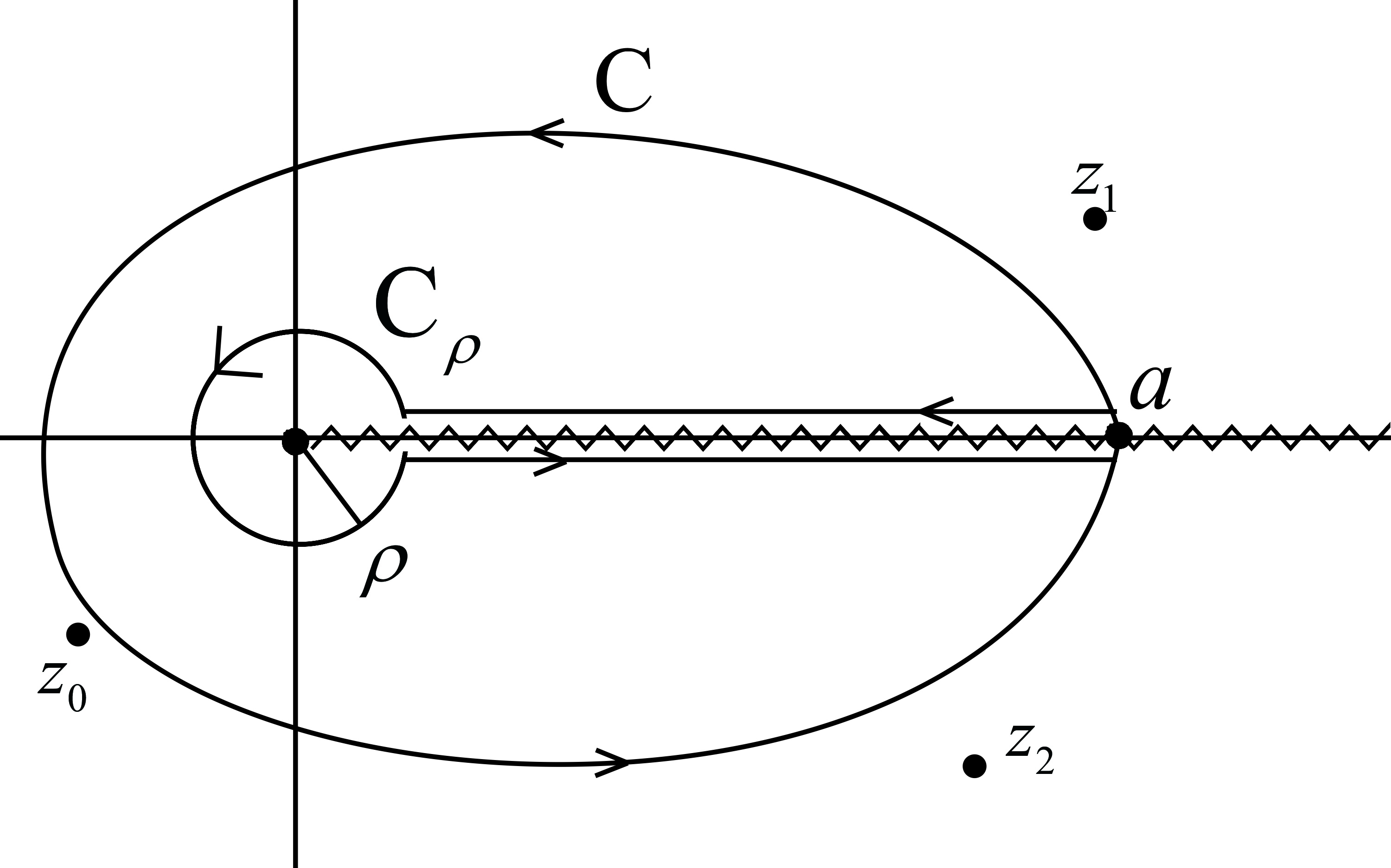}
	\caption{The contour $\mathrm{C}$ in the contour integral representation of the finite part integrals. The points $z_0,\, z_1,\, z_2,\,\dots$ are the singularities of the analytic extension, $f(z)$, of $f(x)$, and none of them is enclosed by $\mathrm{C}$. $z_0$ is the closest singularity to the origin.}
	\label{tear}
\end{figure}

\section{Proof of the legitimacy of the interchange} 
Under the above stated assumption on $f(x)$, the finite part integrals given by equations \eqref{fpi1} and \eqref{fpi2} have the fundamental property that they posses the following contour integral representations in the complex plane, 
	\begin{equation}\label{intergralrep1}
	\bbint{0}{a}\frac{f(x)}{x^{n+1}}\mathrm{d}x=\frac{1}{2\pi i}\int_{\mathrm{C}} \frac{f(z)}{z^{n+1}} \left(\log z-\pi i\right)\mathrm{d}z, \;\; n=0, 1, \dots 
	\end{equation}
	\begin{equation}\label{integralrep2}
	\bbint{0}{a} \frac{f(x)}{x^{m+\nu}}\mathrm{d}x = \frac{1}{\left(\mathrm{e}^{-2\pi \nu i } -1\right)} \int_{\mathrm{C}} \frac{f(z)}{z^{m+\nu}} \mathrm{d}z, \;\;\; m=1, 2, \dots, \;\;\; 0<\nu<1,
\end{equation}
	where $\log z$ and $z^{-\nu}$ take the positive real axis as their branch cuts, and the contour $\mathrm{C}$ is the contour straddling the cut starting from $a$ and ending at $a$ itself, as depicted in Figure-1, with all singularities of $f(z)$ lying outside of $\mathrm{C}$. 
	
	The integral representations \eqref{intergralrep1} and \eqref{integralrep2} were the keys to the evaluation of the Stieltjes integral in terms of the finite parts and to the recovery of the missing terms. And they, too, are the keys in establishing the validity of the interchange of the limit and summation. The proof will proceed in two steps. First is to show the existence of the limit of the finite part integral as $a\rightarrow\infty$ under the assumption that the Stieltjes transform exists. And second is to show the uniform convergence of the infinite series in equations \eqref{stel1} and \eqref{stel2} for sufficiently large $a$. Once uniform convergence is established, the validity of the interchange will follow from the well-known Dominated Convergence Theorem.

\subsection{Case $\nu=0$} First let us establish the necessary condition for the existence of the limit of the finite part as $a\rightarrow\infty$. For some positive $\rho<a$ we deform the contour of integration $C$ in equation \eqref{intergralrep1} into the keyhole contour with the circular portion having the radius $\rho$ as shown in the Figure-1. Then along the deformed contour, the finite part assumes the form,
\begin{equation}\label{deformed1}
\bbint{0}{a} \frac{f(x)}{x^{j+1}}\, \mathrm{d}x = \frac{1}{2\pi i} \int_{\mathrm{C}_{\rho}} \frac{f(z)}{z^{j+1}} \left(\log z - i\pi\right) \, \mathrm{d}z + \int_{\rho}^{a} \frac{f(x)}{x^{j+1}}\, \mathrm{d}x .
\end{equation}
Since the first term is independent of $a$, only the second term determines the behavior in the limit of arbitrarily large $a$. Clearly the limit exists for all non-negative integer $j$ when $x^{-1}f(x)$ is integrable at infinity; but this is one condition to impose in order for the Stieltjes transform to exist. Hence under the assumption that the Stieltjes transform exists, the limit 
\begin{equation}\label{fpilimit}
\lim_{a\rightarrow\infty} \bbint{0}{a} \frac{f(x)}{x^{j+1}}\,\mathrm{d}x
\end{equation}
exists for all non-negative integer $j$. (A similar consideration on the existence of the limit \eqref{fpilimit} has been made in \cite{ticagalapon2018}, but the results there and here have a subtle difference. There the limit was considered without reference to the existence of the Stieltjes integral, in contrast to our intension here which is to show the existence of the limit when the Stieltjes transform exists.)

 Now we show that the infinite series in the right hand side of equation \eqref{stel1} converges uniformly for sufficiently large $a$.  To accomplish this, we establish a bound for the finite part using equation \eqref{deformed1}. Let $a_0$ and $\rho$ be positive numbers such that $\rho<\mathrm{min}\{a_0,|z_0|\}$, where $\rho$ is again the radius of the circle $\mathrm{C}_{\rho}$ in Figure-1; moreover, let the upper limit of integration be $a>a_0$. First let us bound the portion of the integral along the circular path $\mathrm{C}_{\rho}$. Since $f(z)$ is assumed to be analytic at the origin, it admits the expansion
\begin{equation}\label{expansion}
f(z)=\sum_{k=0}^{\infty} c_k z^k,
\end{equation}
with a a radius of convergence $r=|z_0|$.
Substituting the parametrization $z=\rho e^{i\theta}$ in the integral, we obtain
\begin{equation}
\frac{1}{2\pi i} \int_{\mathrm{C}_{\rho}} \frac{f(z)}{z^{j+1}} \left(\log z - i\pi\right) \, \mathrm{d}z = \frac{1}{\rho^j}\left(c_0 \ln\rho + \sum_{k=0,k\neq j}^{\infty} \frac{c_k}{k-j} \rho^k\right),
\end{equation}
from which obtain the desired bound
\begin{equation}
\left|\frac{1}{2\pi i} \int_{\mathrm{C}_{\rho}} \frac{f(z)}{z^{j+1}} \left(\log z - i\pi\right) \, \mathrm{d}z\right| \leq \frac{1}{\rho^j}\left(|c_0 \ln\rho | + \sum_{k=0}^{\infty} |c_k| \rho^k\right)<\infty,
\end{equation}
the right hand side being finite because $\rho$ is less than the radius of convergence of the infinite series. On the other hand, the second term has the finite bound
\begin{equation}
\left|\int_{\rho}^{a} \frac{f(x)}{x^{j+1}}\,\mathrm{d}x\right|\leq \frac{1}{\rho^j} \int_{\rho}^{a}\frac{\left|f(x)\right|}{x}\,\mathrm{d}x\leq \frac{1}{\rho^j} \int_{\rho}^{\infty}\frac{\left|f(x)\right|}{x}\,\mathrm{d}x < \infty ,
\end{equation}
following from the assumed existence of the Stieltjes transform. Combining the two bounds, the finite part has the bound
\begin{eqnarray}
\left|\bbint{0}{a} \frac{f(x)}{x^{j+1}}\mathrm{d}x\right| &\leq& \left|\frac{1}{2\pi i} \int_{\mathrm{C}_{\rho}} \frac{f(z)}{z^{j+1}} \left(\log z - i\pi\right) \, \mathrm{d}z\right| + \left|\int_{\rho}^{a} \frac{f(x)}{x^{j+1}}\, \mathrm{d}x\right| \nonumber \\
&\leq & \frac{1}{\rho^j} \left(|c_0 \ln\rho | + \sum_{k=0}^{\infty} |c_k| \rho^k + \int_{\rho}^{\infty} \frac{|f(x)|}{x} \,\mathrm{d}x\right)<\infty .
\end{eqnarray}
Observe that the right hand side is independent of $a$ for all $a>a_0$ so that it is uniformly bounded for sufficiently large $a$.

We now establish the uniform convergence of the infinite series for all $a>a_0$. The infinite series can be bounded as follows,
\begin{eqnarray}\label{inequality1}
\left|\sum_{j=0}^{\infty} (-1)^j \omega^j \bbint{0}{a} \frac{f(x)}{x^{j+1}}\, \mathrm{d}x \right| &\leq & \sum_{j=0}^{\infty} \omega^j \left|\bbint{0}{a} \frac{f(x)}{x^{j+1}}\, \mathrm{d}x\right| \nonumber \\
&\leq & \sum_{j=0}^{\infty} \omega^j \frac{1}{\rho^j} \left(|c_0 \ln\rho | + \sum_{k=0}^{\infty} |c_k| \rho^k + \int_{\rho}^{\infty} \frac{|f(x)|}{x} \,\mathrm{d}x\right) \nonumber \\
& = & \frac{\rho}{\rho-\omega} S(\rho)<\infty, \;\;\; \omega<\rho<|z_0|,
\end{eqnarray}
where $S(\rho)$ is the quantity in brackets in the second line. Since for every fixed $\omega<|z_0|$ there exists a $\rho$ such that $\omega<\rho<|z_0|$, inequality \eqref{inequality1} implies that the infinite series converges uniformly for a fixed $\omega<|z_0|$ for all $a>a_0$. Then we can interchange the order of limit and summation, and obtain the desired result,
	\begin{equation}\label{stel1limit}
\int_0^{\infty} \frac{f(x)}{\omega+x}\, \mathrm{d}x =  \sum_{j=0}^{\infty} (-1)^j \omega^j \bbint{0}{\infty} \frac{f(x)}{x^{j+1}}\, \mathrm{d}x -f(-\omega) \ln\omega , \;\;\; 0< \omega< |z_0| .
\end{equation}
If it happens that the complex extension, $f(z)$, of $f(x)$ is entire so that $|z_0|=\infty$, the series absolutely converges for all $\omega>0$.  The absolute convergence of the infinite series in equation \eqref{stel1limit} for positive $\omega<|z_0|$ extends the validity of equation \eqref{stel1limit} to complex arguments by analytic continuation. Equation \eqref{stel1limit} is valid for complex $\omega$ with $|\mathrm{arg}\,\omega|<\pi$, with $\ln\omega$ assuming its principal value. 

\subsection{Case $\nu\neq 0$} Now for the second case, the finite part integral assumes the form 
\begin{equation}\label{branchexpand}
\bbint{0}{a} \frac{f(z)}{x^{j+\nu +1}} \mathrm{d}x= \frac{1}{e^{-2\pi \nu i}-1}\int_{\mathrm{C}_{\rho}} \frac{f(z)}{z^{j+\nu+1}} \mathrm{d} z + \int_{\rho}^{a} \frac{f(x)}{x^{j+\nu+1}} \mathrm{d} x 
\end{equation}
after deformation of the contour $\mathrm{C}$ to the keyhole contour as depicted again in Figure-1. In this form it is sufficient for $x^{-\nu-1} f(x)$ to be integrable at infinity in order for the limit as $a\rightarrow\infty$ to exist.  Again this is the same condition for the Stieltjes transform to exist. So, under the existence of the Stieltjes transform, the limit exists.

We now show that the summation and the limit can be interchanged for $\nu\neq 0$ in the same manner that we did above for $\nu=0$. Under similar conditions for $\rho$, $a_0$ and $a$, we establish uniform convergence of the infinite series for sufficiently large $a$. First, let us obtain a bound for the finite part integral for a given fixed finite $a$.  For the first term of equation \eqref{branchexpand}, we have
\begin{equation}\label{left}
\frac{1}{e^{-2\pi \nu i}-1}\int_{\mathrm{C}_{\rho}} \frac{f(z)}{z^{j+\nu+1}} \mathrm{d} z = \frac{1}{\rho^{j+\nu}} \sum_{k=0}^{\infty} \frac{c_k}{(k-j-\nu)} \rho^k ,
\end{equation}
on substituting the expansion \eqref{expansion} back into the left hand side of equation \eqref{left}. From which we have the bound
\begin{equation}
\left|\frac{1}{e^{-2\pi \nu i}-1}\int_{\mathrm{C}_{\rho}} \frac{f(z)}{z^{j+\nu+1}} \mathrm{d} z \right|\leq  \frac{1}{\rho^{j+\nu} \nu} \sum_{k=0}^{\infty} \left|c_k\right| \rho^k ,
\end{equation}
for $\rho<|z_0|$. On the other hand, the second term has the bound,
\begin{equation}
\left|\int_{\rho}^{a} \frac{f(x)}{x^{j+\nu+1}} \mathrm{d} x \right| \leq \frac{1}{\rho^j} \int_{\rho}^{a} \frac{\left|f(x)\right|}{x^{\nu+1}} \mathrm{d} x \leq \frac{1}{\rho^j} \int_{\rho}^{\infty} \frac{\left|f(x)\right|}{x^{\nu+1}} \mathrm{d} x .
\end{equation}
Under the assumption that the Stieltjes transform exists, the most right hand side is finite. Combining the two bounds, we obtain the following bound for the finite part integral,
\begin{equation}\label{bound2}
\left|\bbint{0}{a} \frac{f(x)}{x^{j+\nu +1}} \mathrm{d}x\right| \leq \frac{1}{\rho^j}\left(\frac{1}{\nu \rho^{\nu}} \sum_{k=0}^{\infty}|c_k| \rho^k+ \int_{\rho}^{\infty} \frac{\left|f(x)\right|}{x^{\nu+1}} \mathrm{d} x\right) <\infty .
\end{equation}
The right hand side is also independent of $a$ for all $a>a_0$ so that the finite part for the present case is uniformly bounded for sufficiently large $a$.

Given the inequality \eqref{bound2} for the finite part, we obtain the following bound for the infinite series, 
\begin{eqnarray}
\left| \sum_{j=0}^{\infty} (-1)^j \omega^j \bbint{0}{a}\frac{f(x)}{x^{j+\nu+1}}\, \mathrm{d}x\right| &\leq & \sum_{j=0}^{\infty} \omega^j\left|\bbint{0}{a} \frac{f(x)}{x^{j+\nu +1}} \mathrm{d}x\right| \nonumber \\
&\leq &\sum_{k=0}^{\infty} \frac{\omega^j}{\rho^j} \left(\frac{1}{\nu \rho^{\nu}} \sum_{k=0}^{\infty}|c_k| \rho^k+ \int_{\rho}^{\infty} \frac{\left|f(x)\right|}{x^{\nu+1}} \mathrm{d} x\right) \nonumber \\
&=& \frac{\rho}{\rho-\omega} R(\rho) <\infty, \;\;\; \omega < \rho < |z_0|,
\end{eqnarray}
where $R(\rho)$ is the quantity in parenthesis in the second line. Again the bound is independent of $a$ for all $a>a_0$ so that the infinite sum is uniformly convergent for sufficiently large $a$ for every fixed $\omega<|z_0|$. Then we can interchange the summation and limit to obtain the desired result,
\begin{equation} \label{stel2limit}
\int_0^{\infty} \frac{x^{-\nu} f(x)}{\omega + x} \mathrm{d}x = \sum_{j=0}^{\infty} (-1)^j \omega^j \bbint{0}{\infty}\frac{f(x)}{x^{j+\nu+1}}\, \mathrm{d}x + \frac{\pi f(-\omega)}{\omega^{\nu} \sin\pi\nu} ,\;\;\; \omega<|z_0|,
\end{equation}
for $0<\nu<1$. Again by analytic continuation the validity of equation \eqref{stel2limit} can be extended to complex $\omega$ with $|\mathrm{arg}\,\omega|<\pi$ and with $\omega^{\nu}$ assuming its principal value. 

\section*{Acknowledgment}
The author acknowledges the Office of the Chancellor of the University of the Philippines Diliman, through the Office of the Vice Chancellor for Research and Development, for funding support through the Outright Research Grant  171711 PNSE.